\documentclass{article}
\usepackage{cite,tikz,mdframed}
\usepackage{amsmath,amsthm,amssymb,tikz,mdframed,bookmark}

\newtheorem{definition}{Definition} 
\newtheorem{lemma}{Lemma}
\newtheorem{example}{Example}
\newtheorem{theorem}{Theorem}
\newtheorem{corollary}{Corollary}
\newtheorem{observation}{Observation}
\newcommand{\AND}{\wedge}
\newcommand{\OR}{\vee}

\usepackage{array}
\usepackage{url}

\title{Kernelization, Proof Complexity and Social Choice} 

\author{Gabriel Istrate, Cosmin Bonchi\c{s} and Adrian Cr\~{a}ciun\thanks{West University of Timi\c{s}oara, Romania. Corresponding author email: gabrielistrate@acm.org}}

\begin{document}
\maketitle

\begin{abstract}
We display an application of the notions of kernelization and data reduction from parameterized complexity to proof complexity: Specifically, we show that the existence of data reduction rules for a parameterized problem having (a). a small-length reduction chain, and (b). small-size (extended) Frege proofs certifying the soundness of reduction steps  implies the existence of subexponential size (extended) Frege proofs for propositional formalizations of the given problem. 

We apply our result to infer the existence of subexponential Frege and extended Frege proofs for a variety of problems. Improving earlier results of Aisenberg et al. (ICALP 2015), we show that propositional formulas expressing (a stronger form of) the Kneser-Lov\'asz Theorem have polynomial size Frege proofs for each constant value of the parameter $k$. Previously only quasipolynomial bounds were known (and only for the ordinary Kneser-Lov\'asz Theorem).  

Another notable application of our framework is to impossibility results in computational social choice: we show that, for any fixed number of agents, propositional translations of the Arrow and Gibbard-Satterthwaite impossibility theorems have subexponential size Frege proofs. 
\end{abstract}


% \fi
%
\section{Introduction} 

The central task of \emph{proof complexity} \cite{buss2019proof,krajivcek2019proof} is that of understanding (and distinguishing) the relative power of various propositional proof systems. Proving lower bounds for stronger and stronger proof systems might (in principle) be a way to eventually confirm the various conjectures of computational complexity. Yet we are far from being able to prove exponential lower bounds for some concrete problems in strong proof systems. 

One of the most important open problems in this area, explicitly raised by Bonet, Buss and Pitassi  \cite{bonet1995there} is that of separating the complexity of Frege proof systems ("textbook style propositional proofs") from that of extended Frege proof systems (which in addition can introduce new variables as substitutes for arbitrary propositional formulas). That is, we would like to find explicit classes of propositional formulas that have extended Frege proofs of polynomial size  but have exponential lower bounds on the size of the shortest Frege proofs. 

Many classes of problems that are candidates for separating the two systems have been proposed, e.g.  statements based on linear algebra \cite{soltys2004proof,hrubes2015short}, propositional encodings of the Paris-Harrington independence results \cite{carlucci2011paris}, Ramsey's theorem \cite{lauria2017complexity}, central theorems from extremal combinatorics \cite{nozaki2008polynomal,aisenberg2016quasipolynomial}, or the Kneser-Lov\'asz formula from combinatorial topology \cite{istrate2014proof,aisenberg2018short}).

 So far most of the proposed examples have turned out to have sub-exponential Frege proofs (only a couple of candidate formula classes for the purported separation have been advanced, such as local improvement principles \cite{kolodziejczyk2011provably}, or truncations of the octahedral Tucker lemma \cite{aisenberg2018short}). On the other hand many of the tractability results listed above have  been obtained using techniques that are highly problem-specific, with relatively little transferability to more general classes of formulas. The existence of such general methods would be highly desirable: such general results could guide the search for examples witnessing the desired separation by pointing to structural properties one needs to avoid in order to construct them. 
  
The purpose of this paper is to present a more general approach for proving sub-exponential upper bounds for the Frege and extended Frege proof complexity of some classes of propositional formulas. We point out that concepts from the theory of parameterized complexity \cite{downey2013fundamentals},   specifically those of \emph{data reduction} and \emph{kernelization}  \cite{fomin2019kernelization} may be relevant to proof complexity as well\footnote{This is not the first time a connection between parameterized complexity and proof complexity was made; see e.g. \cite{dantchev2011parameterized,beyersdorff2012parameterized}. However, our concerns are rather different.}. We give a metatheorem that translates a data reduction for the original problem whose soundness can be witnessed by polynomial size (extended) Frege proofs into subexponential proofs for the corresponding propositional translation of the original problem. The exact size of these proofs is controlled by three factors: the length of the data reduction chain, the nature and size of the proofs witnessing the soundness of the reduction rules, and the size of proofs of unsatisfiability for the formulas in the kernel. 

We give several applications of our metatheorem. The use of kernelization techniques does not only allow to tackle the complexity of new problems, but also to improve existing results: In  \cite{aisenberg2018short} it was shown that propositional formulas $Kneser_{n,k}$ expressing a principle from topological combinatorics known as 
\emph{the Kneser-Lov\'asz theorem} have, for every fixed value of parameter $k$, quasipolynomial size Frege proofs. \textbf{We improve this result in several ways:} first, our result deals with a result stronger than the Kneser-Lov\'asz theorem, known as Schrijver's theorem. Second, \textbf{we get polynomial (instead of quasipolynomial) upper bounds for Frege proofs}. Other applications of our metatheorem concern several  (mostly graph-theoretic) problems whose kernelization had previously been studied in the theory of parameterized algorithms. The problems we study are well-known examples satisfying two conditions: First, their negative instances have natural formulations as unsatisfiable CNF formulas. Second, they have efficient kernelizations, often with a small kernel.  The problems we have chosen illustrate an important point: \emph{several techniques used in the literature to prove the existence of a kernelization can often be efficiently simulated by (extended) Frege proofs.} Perhaps the most interesting set of applications of our general metatheorem comes, however, from the theory of computational social choice \cite{brandt2016handbook}: as it was recently observed, various propositional formalizations of impossibility principles in the theory of computational social choice have computer-assisted proofs that reduce the task of mathematically proving these theorems to the verification of a finite number of cases of unsatisfiability of propositional formulas (using SAT solvers; see \cite{geist2017computer} for a fairly recent overview)\footnote{A similar phenomenon had been independently uncovered for the Kneser-Lov\'asz theorem \cite{aisenberg2018short}.}. \textbf{We obtain subexponential upper bounds on the complexity of Frege proofs for propositional formulations of  Arrow's theorem and the Gibbard-Satterthwaite theorem: quasipolynomial in general, polynomial for a fixed number of agents.}

The outline of the paper is as follows: in Section~\ref{sec:prelim} we review some useful notions and concepts. Our main result is stated and proved in Section~\ref{sec:main}. We then apply it to a couple of problems in Combinatorial Topology and Parameterized Algorithms. Section~\ref{sec:choice} discusses applications to Computational Social Choice. We conclude with several discussions and open problems. 
 \textbf{Unless stated otherwise, proofs are given in the Appendix.} Also, the logical formalizations of results claiming polynomial size (extended) Frege proofs are only sketched. 
\section{Preliminaries}
\label{sec:prelim} 
We assume basic familiarity with concepts from three distinct areas: proof complexity, parameterized algorithms and computational social choice. We refer the reader to \cite{krajivcek2019proof, fomin2019kernelization,brandt2016handbook} for book-length treatments of these topics. Nevertheless, for  purposes of readability we review a couple of relevant notions in the sequel: 

\begin{definition}
  \emph{Frege proof systems} are sound and complete propositional proof systems having a finite number of axioms and inference rules. All Frege proof systems are equivalent up to polynomial transformations \cite{cook1979relative}. Therefore, for concreteness, we will employ a standard "textbook proof style" system having \emph{modus ponens} as the unique inference rule. 
  
 An \emph{extended Frege proof} augments Frege proofs by allowing  new variables to substitute  complex formulas. In both cases we measure the length of a proof by the number of clauses in it. Thus the effect of the extension rule in extended Frege proofs is reducing proof length. 
\end{definition} 

\newpage
\begin{mdframed} 
When we will refer to the proof complexity of witnessing an implication $\Phi_{1}\vdash \Phi_{2}$ using (extended) Frege proofs, what we mean is that there are distinct sets of variables $X,Y$ for $\Phi_{1},\Phi_{2}$ and substitutions $Y_{i}=\Xi_{i}[X], i=1,2$ such that one can derive the clauses of $\Phi_{2}[Y]$ using the clauses of $\Phi_1[X]$ as axioms. Since they use substitutions, these are extended Frege proofs. To convert them into Frege proofs one needs to expand the definitions of new variables.  
\end{mdframed}

We use the shorthand $[m]$ for the set $\{1,2,\ldots, m\}$, $[i:j]$ for $\{i,i+1,\ldots, j\}$, and write $A\cong B$ when sets $A,B$ have the same cardinal. Function  $f(\cdot)$ is called \emph{quasipolynomial} if there exists $k>0$ such that $f(n)=O(2^{O(\log^{k}(n)})$. 
We will need the following simple

\begin{lemma} 
Suppose $C$ is a CNF formula and $Z_{1},\ldots, Z_{m}$ are literals s.t. $C\AND (Z_{1}\wedge Z_{2}\wedge \ldots \wedge Z_{m})$ is unsatisfiable, as witnessed by a resolution (Frege) proof of length $k$. Then one can derive from $C$ clause $\overline{Z_{1}}\vee  \overline{Z_{2}}\vee \ldots \vee \overline{Z_{m}}$ via a resolution (Frege) proof of size at most $k$. 
\label{foo-size} 
%\textcolor{red}{Is this correct for Frege/extended Frege ?}
\end{lemma}

%\subsection{Concepts of parameterized complexity: parameterizations, data reductions, kernelizations}

\begin{definition}[Parameterized problem] 

 Let $\Sigma$ be an alphabet. $L$ is a \emph{parametrized problem} over $\Sigma^*$ iff $L \subseteq \Sigma^* \times \mathbb{N} .$ Define \textit{the support of $L$},  by $
supp(L)=\{x\in \Sigma^* | (\exists k\in \mathbb{N}): (x,k)\in L\}.$
\end{definition}

Let $L$ be a parameterized problem in co-NP. Let $\phi$ be a "canonical" reduction of $L$ to $\overline{SAT}$. When $\phi$ is clear from the context, we identify $L$ with the set of pairs 
$\phi(L):=\{(\phi(x,k),k):(x,k)\in L\}$, slightly abusing notation, and writing $L$ instead of $\phi(L)$. 

\begin{example}[Graph colorability]
 \textit{ } Let $ COL = \left\{ (G, i) \ \middle| \ \chi(G)\leq i \right\}.$
We can encode instances $(G,k)$ of $COL$ as SAT instances $(\phi(G,k),k)$ by the reduction $\phi$ informally defined by: 
\begin{description} 
\item[-] For $v\in V(G)$ and $1\leq i\leq k$ define boolean $X_{v,i}=$TRUE iff $v$ is colored with color $i$. 
\item[-] For every sets of distinct vertices $v,w\in G$ we define variable $Y_{v,w}$. The semantics is that $Y_{v,w}=TRUE$ means that $v$ and $w$ are connected by an edge. Thus, for all sets $\{v,w\}$ that correspond to an edge we add to $\phi(G,k)$ the unit clause $Y_{v,w}$. On the other hand, for sets $\{v,w\}$ that correspond to non-edges we add to $\phi(G,k)$ the unit clause $\overline{Y_{v,w}}$.
\item[-] For every $v\in V$  add $X_{v,1}\vee X_{v,2} \vee \ldots \vee X_{v,k}.$ ($v$ must be colored with one of colors $1$ to $k$")
\item[-] For $v\in V$ and $1\leq i<j\leq k$ add  $\overline{X_{v,i}}\vee \overline{X_{v,j}}.$ ("$v$ cannot be colored with both $i$ and $j$")
\item[-] For every set $v,w\in V$ and $i\in 1\ldots k$, add clause $\overline{Y_{v,w}}\vee \overline{X_{v,i}}\vee \overline{X_{w,i}}.$ ("if v and w are connected then they cannot both be colored with color $i$")
\end{description} 
\label{ex:col}
\end{example}

\begin{definition}
The Kneser-Lov\'asz theorem (see e.g. \cite{de2012course}) is a statement about the chromatic number of the following graph, $Kn_{n,k}$, parameterized by an integer $k\geq 1$: The vertex set of $Kn_{n,k}$ is ${{n}\choose {k}}$, the set of subsets of $\{1,2,\ldots, n\}$ with $k$ elements. Two sets $A,B$ represent adjacent vertices iff $A\cap B=\emptyset$. 
The Kneser-Lov\'asz theorem can be equivalently restated as $\chi(Kn_{n,k})>n-2k+1$\footnote{actually $\chi(Kn_{n,k})=n-2k+2$. However, the existence of a ($n-2k+2$)-coloring is easy \cite{de2012course}.}. % and won't concern us.}. 
It is expressed as a parameterized problem as follows: 
\[
L_{Kn} = \left\{({Kn}_{n}^k, i):  %\ \middle| \ 
 n\geq 2k >1,\ i \leq n-2k+1 \right\}.
\]
\end{definition}

Note that $L_{Kn} \subseteq \overline{COL}$, hence we can use the translation from Example~\ref{ex:col} to canonically translate $L_{Kn}$ as a set of unsatisfiable propositional formulas.   

The next problem is just the graph coloring problem, but with a different parameterization: 

\begin{definition} 
An instance of the Dual Coloring problem is a pair $(G,k)$, where $G$ is a graph with $n$ vertices and $k$ is an integer. 
To decide: is $\chi(G)\leq n-k$ ? That is, let
$DualCol=\{(G,k):\chi(G)\leq n-k \}$. 
We have $(G,k)\in DualCOL \Leftrightarrow (G,n-k) \in COL.$
For this reason the translation of $\overline{DualCOL}$ into $\overline{SAT}$ modifies the one from $\overline{COL}$ to  $\overline{SAT}$ in Example~\ref{ex:col} in an obivious way. Note also that $\overline{DualCol}$ also generalizes the Kneser-Lov\'asz theorem, since the harder part of this theorem is equivalent to $(Kn_{n,k},2k-1)\in \overline{DualCol}$.
\end{definition}

Given graph $G$, a \emph{vertex cover in $G$} is a set $S\subseteq V(G)$ such that for every edge $e=(v,w)$, $v\in S$ or $w\in S$. We denote by $vc(G)$ the size of the smallest vertex cover of $G$.

\begin{example}[Vertex Cover]
 \textit{ } Let $ \overline{VC} = \left\{ (G, i) \ \middle| \ i < vc(G) \right\}$ be the set of unsatisfiable instances of Vertex Cover. We can encode (negative) instances $(G,k)$ of $VC$ as  instances $\phi(G,k)$ of $SAT$ by the reduction $\phi$ informally defined as follows:
\begin{description} 
\item[-] For every $v\neq w\in V$, $(v,w)\in E$ add new unit clause $Y_{v,w}$ to the formula. For $(v,w)\not \in E$ add new unit clause $\overline{Y_{v,w}}$, to the formula. 
\item[-] For $v\in V(G)$ and $i\in 1\ldots k$ define boolean variable $X_{v,i}$ with the informal semantics $X_{v,i}$ is TRUE when vertex $v$ is the $i$'th vertex in a vertex cover of size $k$. To encode this semantics add to the formula, for every $v\in V$ and $i\in 1,\ldots, k$, clause
$\overline{X_{v,i}}\vee (\bigvee_{w\neq v} Y_{v,w})$. This ensures that if $v$ is chosen in the vertex cover then it covers some edge $(v,w)$. With some extra technical complications one can do away with adding these clauses. 
\item[-] For every $i=1,\ldots, k$ we add to the formula clause $\bigvee_{v\in V} X_{v,i}.$
\item[-] For every $v\neq w\in V$ and $1\leq i \leq k$ we add to the formula clause $\overline{X_{v,i}}\vee \overline{X_{w,i}}.$
\item[-] For every $v\in V$ and $1\leq i <j \leq k$ we add to the formula clause $\overline{X_{v,i}}\vee \overline{X_{v,j}}.$

\item[-] For  $v\neq w\in V$ add to the formula clause $\overline{Y_{v,w}}\vee X_{v,1}\vee \ldots \vee X_{v,k}\vee X_{w,1}\vee \ldots \vee X_{w,k}.$
\end{description} 
\end{example}

\begin{definition}[Kernelization]

 Let $L$ be a parametrized problem. A \emph{kernelization algorithm} (or, shortly, kernelization) $\mathcal{\text{Ker}}$ for the problem $L$ is an algorithm that works as follows: on input  $(x, k)$, $Ker$ outputs (in time polynomial in $|(x,k)|$ ) a  pair $(x', k'),$ such that the following are true: $(x,k) \in L \text{ iff } (x', k') \in L$, and $|x'|, k' \leq g(k)$, where $g$ is a computable function. 
 Pair $(x^{\prime},k^{\prime})$ is called \emph{the kernel of $(x,k)$}, while $g(k)$ is called \emph{the size of the kernel.} 
\end{definition}

One can convert a kernelization into an algorithm by solving kernel instances by other means (e.g. brute force). A kernelization is often the reflexive, transitive closure of a finite set of \emph{data reduction rules}: we apply the  rules as long as possible, until we are left with an instance, the kernel, to which no rule can be applied anymore.

\begin{definition}[Data reduction rule] %\textit{}

 Let $L$ be a parameterized problem. A \textbf{data reduction rule for $L$} is an algorithm $\mathcal{A}$ that maps (in time polynomial in $|x|+k$) an instance 
  $(x, k)$ of $L$ to an instance $(x', k')$ such that
 $(x,k) \in L \text{ iff } (x', k') \in L$ (we say that the two instances are \emph{equivalent}, or that the reduction rule is \emph{safe}), and $|x'|\leq |x|$. 
 In practice, a data reduction rule may be well-defined only for $|x| \geq f(k)$, for some function $f(\cdot)$, as we can simply extend it to smaller instances $(x,k)$ by defining $A(x,k)=(x,k)$. All kernelizations in this paper have this nature, and we will assume this to be true for all the results we give in the sequel. 
\end{definition}

\begin{definition}[Data reduction chain] Given parameterized problem $L$ kernelizable via data reductions $(A_{1},A_{2}, \ldots, A_{r})$, a \textbf{data reduction chain} for instance $(x,k)$ of $L$ is a sequence 
$(x_{0},k_{0}), (x_{1},k_{1}), \ldots, (x_{m},k_{m})$, where $(x_{0},k_{0})=(x,k)$, $A_{t}(x_{m},k_{m})=(x_{m},k_{m})$, for all $t=1,\ldots r$ and, for all $i=1,\ldots, m$ there exists $j\in 1,\ldots, r$ such that $(x_{i},k_{i})=A_{j}(x_{i-1},k_{i-1})$. 
\end{definition} 

\begin{example}[Data reduction for  Kneser instances:]

\begin{description}
 \item[] Reductions $(Kn_n^2, a) \xrightarrow{ } (Kn_{n-1}^2, a-1)$ and $(Kn_n^3, a) \xrightarrow{ } (Kn_{n-1}^3, a-1)$  were used in \cite{istrate2014proof} to give polynomial size extended Frege upper bounds for Kneser formulas for $k=2,3$. 
 \item[-] For $k\geq 2$ there exists $N(k)\leq k^4$ such that for $n>N(k)$ $ (\mathcal{K}_n^k, a) \xrightarrow{ } (\mathcal{K}_{n-1}^k, a-1)$. This was used in \cite{aisenberg2018short} to give polynomial size extended Frege upper bounds for Kneser formulas. 
 \item[-] For $k\geq 2$ there exists $N(k)\leq k^4$ such that for $n>N(k)$ $ (Kn_n^k, a) \xrightarrow{ }  (Kn_{n-\tfrac{n}{2k}}^k, a-\tfrac{n}{2k})$. This was used in \cite{aisenberg2018short} to give quasipolynomial size Frege proofs for Kneser formulas. 
\end{description}

\end{example}

\begin{definition} 
A \textbf{crown decomposition of a graph $G$} (see e.g. Fig. \ref{fig:kneser} b.) is a decomposition of $V(G)$ into three subsets $C,H,R$, 
$C\neq \emptyset$ such that 
(1). $C$ is an independent set. (2). No vertex in $C$ is adjacent to a vertex in $R$. 
(3). There exists a matching of $H$ in $C$, i.e. a set of disjoint edges covering $H$ with the other endpoint in $C$.
\label{cd} 
\end{definition}

%\subsection{Concepts from theory of social choice} 

Given a set $S$ and $T\subseteq S$, we will denote by $S_{-T}$ the set $S\setminus T$. We will also write $S_{-a}$ instead of $S_{-\{a\}}$. 
When $S=[m]$, of course $[m]_{-T}\cong [m-|T|]$ for every $T\subseteq [m]$. 

\begin{definition} 
Given a set of $m$ objects, identified with the set $[m]$, a \emph{preference profile} is a linear ordering of $[m]$, i.e. a permutation $\pi\in S_{m}$. Given $a,b\in [m]$ we say that \emph{$a$ is preferred to $b$} (written $a <_{\pi} b$) iff $\pi^{-1}(a)<\pi^{-1}(b)$.  \textbf{Note that preferred objects are lower in the ordering.} We denote by $top(\pi)$ the object $\pi^{-1}(1)$, i.e. the object that is preferred in $\pi$ to all others. Given a preference profile $\pi$ and $T\subseteq [m]$, denote by $\pi_{-T}$ the restriction of $\pi$ to $[m]_{-T}$, and by $\pi^{+T}$ the preference profile derived from $\pi$ by making all elements $a\in T$ less preferred than any other $b\in [m]$ (with an arbitrary fixed order among them, e.g. the order induced on $[m]$ by the identical permutation).   
\end{definition} 

\begin{definition} Given a set of $m$ objects, identified with the set $[m]$ and a set of $n$ agents, a \textbf{social choice function} (SCF) is a mapping $s:S_{m}^{n}\rightarrow Z$. $Z$ is a set equal to $S_m$ (for Arrow's theorem) and to $[m]$ (for the Gibbard-Satterthwaite theorem).  A SCF is \textbf{dictatorial} if there exists $i\in [m]$ such that for all $R_{1},R_{2},\ldots, R_{n}\in S_{m}$, $s(R_{1},R_{2},\ldots, R_{n})=R_{i}$ ($s(R_{1},R_{2},\ldots, R_{n})=top(R_{i})$ for the Gibbard-Satterthwaite theorem). A SCF is \textbf{unanimous}  if whenever $a$ is preferred to $b$ in all profiles $R_{1},R_{2},\ldots, R_{n}$ then $a$ is preferred to $b$ in profile $s(R_{1},R_{2},\ldots, R_{n})$. SCF $s$ satisfies the \textbf{independence of irrelevant alternatives (IIA) axiom} if whenever $a,b\in [m]$ are two different objects and $(R_{1},R_{2},\ldots, R_{n})\in S_{m}^{n}$ and $(R^{\prime}_{1},R^{\prime}_{2},\ldots, R^{\prime}_{n})\in S_{m}^{n}$ are two vectors of preference profiles such that, for all $i=1,\ldots m$, $R_{i}$ and $R^{\prime}_{i}$ agree in their relative preference of $a$ or $b$, then $s(R_{1},R_{2},\ldots, R_{n})$ and $s(R^{\prime}_{1},R^{\prime}_{2},\ldots, R^{\prime}_{n})$ agree in their relative preference of $a$ or $b$. A SCF is \textbf{onto} iff it is onto as a function. Finally, for every pair $(R,o)$, $R=(R_1,\ldots, R_n)$ and player $1\leq i=1\leq m$, denote by $pr(i,o,R)$ the set of objects $o^{\prime}$ s.t. $R_{i}^{-1}(o)\leq R_{i}^{-1}(o^{\prime})$ (i.e. $i$ weakly prefers $o$ to $o^{\prime}$ in $R_{i}$). A SCF $s$ is \textbf{strategyproof} iff, for every strategy profile $R$, if $o$ is the outcome of preference profile $R$ then $i$ cannot misrepresent its preferences as $\pi \in S_{m},\pi \neq R_{i}$ so that the social choice for the resulting profile $s(i,R,\pi)$ is an $o^{\prime}$ that $i$ strictly prefers to $o$. 
\end{definition}  

Given an $SCF$ $W:[m]^{n}\rightarrow Z$ and $B\subseteq [m]$ we define function $W_{-B}:[m]^{n}_{-B}\rightarrow Z$ to be defined as follows: $$W_{-B}(R_{1},R_{2},\ldots, R_{n})=W(R_{1}^{+B},R_{2}^{+B},\ldots, R_{n}^{+B})_{-B}.$$
 In other words, we extend profiles $R_1,R_2,\ldots, R_n$ by making objects in $B$ less preferred than all other objects, apply $W$ on the resulting profiles, then drop objects from $B$ from the result. 

\section{Main (Meta)Theorem and Applications} 
\label{sec:main} 

%\begin{definition} 
%\end{definition} 
In the next definition we formalize the complexity of simulating data reduction steps by (extended) Frege proofs. Clearly, we want to encode the scenario where each such step can be simulated by efficient proofs. Our main result will allow a slightly more general setting, where the safety of each reduction step can be established by a "case by case argument with a limited number of cases". This will lead not to a chain but to a \emph{tree} of logical reductions: 

%\begin{definition} Given reduction rule $\mathcal{A}$ for problem $L$, we say that \textbf{the soundness of $\mathcal{A}$ has \textbf{polynomial size (extended) Frege proofs}} iff there is a polynomial $p(n)$ so that for every $(x,k)\not \in L$ the formula $\Phi(x,k)\vdash \Phi(x^{\prime},k^{\prime})$
%that witnesses the safety of the statement $(x,k)\not \in L \Rightarrow (x^{\prime},k^{\prime})= \mathcal{A}(x,k)\not \in L$ 
%has (extended) Frege proofs of size at most $p(|\Phi(x,k)|)$. 
%\end{definition} 

\begin{definition} Given reduction rule $\mathcal{A}$ for problem $L$ and function $h(\cdot)$, \textbf{the soundness of $\mathcal{A}$ has \textbf{(extended) Frege proofs} of size $h(\cdot)$} iff there is an  integer $R\geq 1$ s.t. for every $(x,k)\not \in L$ and every step $(x_i,k_i)\rightarrow (x_{i+1},k_{i+1})$ in the reduction chain the following are true: 
\begin{itemize} 
\item There exists $r^{\prime}_{i}\leq R$, tautology $\Xi_i:= \bigvee_{t=1}^{r^{\prime}_{i}} \Xi_{i,t}$  and formulas $\eta_{i,1}, \ldots \eta_{i,r^{\prime}_{i}}$ isomorphic (up to a variable renaming) to $\Phi(x_{i+1},k_{i+1})$ s.t. for $t=1,\ldots, r^{\prime}_{i}$, $
\Phi(x_i,k_i)\AND \Xi_{i,t} \vdash \eta_{i,t}$. 
\item Proving the soundness of $\Xi_i$ and of all  reductions $
\Phi(x_i,k_i)\AND \Xi_{i,t} \vdash \eta_{i,t}$ 
can be accomplished by (extended) Frege proofs of \textbf{total} size at most $h(|\Phi(x,k)|)$.
\end{itemize}  
\end{definition}

Given this definition, our main (meta)theorem is: 

\begin{theorem} 
Let $L$ be a parameterized problem that is  kernelizable via a finite number of data reduction rules $(A_1,A_{2},\ldots, A_{r})$ with kernel size $g(\cdot)$. 
\begin{enumerate} 
\item Assume that negative instance $(x,k)$ of $L$ has a data reduction chains of length $C(x,k)$, and that the soundness of each reduction rule $A_1,A_{2},\ldots, A_{r}$ can be witnessed using extended Frege proofs of size at most $h(|\Phi(x,k)|)$, for some function $h(\cdot)$. Then $L$ has extended Frege proofs of size \[
O((\sum\limits_{i=0}^{C(x,k)} R^{i})[h(|\Phi(x,k) |)+2^{O(poly(g(k)))}]).
\]
In particular, if $R=1$ and for every fixed $k$ we have $C(x,k)=O(poly(|\Phi(x,k)|))$ then, for every fixed $k$, negative instances $\Phi(x,k)$ of $L$ have extended Frege proofs of size polynomial in $|\Phi(x,k)|$. 
\item Assume that negative instances $(x,k)$ of $L$ have data reduction chains of length $C(x,k)=O(1)$ ($O(log(|\Phi(x,k)|))$, respectively), where the constant may depend on $k$, and that the safety of each reduction $\Phi(x_{i},k_{i})\vdash \Phi(x_{i+1},k_{i+1})$  is witnessed by Frege proofs of size $\leq p(|\Phi(x,k)|)$, for some fixed polynomial $p(\cdot)$ Then for every fixed $k$,  negative instances $(\Phi(x,k),k)$ of $L$ have Frege proofs of size polynomial (quasipolynomial) in $|\Phi(x,k)|$.
\end{enumerate} 
\label{thm-main}
\end{theorem} 

%\begin{proof}
%See Appendix. 
%\end{proof} 

Next we highlight some application of our main (meta)theorem: 
\iffalse
first we improve the result in 
 \cite{aisenberg2018short}, by giving efficient proofs for another statement in combinatorial topology that implies the Kneser-Lov\'asz theorem: Schrijver's theorem \cite{schrijver1978vertex}, the "stable" version of the Kneser-Lov\'asz theorem. Then we point out that many of the existing kernelizations from the literature can be witnessed by efficient (extended) Frege proofs, hence implying efficient proofs for negative instances of a variety of problems: Vertex Cover, (dual) graph coloring, $d$-Hitting Set. The main point is that various proof techniques (ad-hoc methods for Vertex Cover and Edge Clique Cover, \emph{crown decomposition} for dual graph coloring, the \emph{sunflower lemma} of Erd\H{o}s-Rad\'o for $d$-Hitting set) lead to efficient propositional proofs. These techniques have been applied to a variety of problems (see \cite{fomin2019kernelization}), so similar results probably hold as well for other problems to which these techniques apply. Finally, we exploit existing kernelizations from the literature to give efficient Frege proofs for propositional restatements of two famous results from the social choice literature: Arrow's theorem \cite{arrow2012social}, and the Gibbard-Satterthwaite theorem \cite{gibbard1973manipulation,satterthwaite1975strategy}. We obtain partial results, in the form of efficient \emph{extended} Frege proofs. It is an interesting open problem whether similar results hold for Frege proofs too. 
\fi 
\subsection{Proof Complexity of (Dual) Coloring}

\begin{theorem} 
There exists a kernelization that reduces  instances $(G,k)$ of DUALCOL to a kernel of size at most $3k-2$. The length of reduction chain in this kernelization is $O(k)$. The soundness of each reduction step can be witnessed by polynomial size Frege proofs. Hence, for every fixed $k$, negative instances $(G,k)$ of DUALCOL have Frege proofs of size polynomial in $|\Phi_{G,k}|$. 
\label{theorem-dualcol}
\end{theorem} 

\begin{proof} 

The kernelization is a variant of the classical one from the parameterized complexity literature, based on  \emph{crown decompositions} (Definition~\ref{cd}). It consists of three data reductions: 
\begin{itemize} 
\item[(a).] Let $All(G)$ be the set of vertices $v$ adjacent to all other vertices in $G$.  If $All(G)\neq \emptyset$ then $(G,k)\in DualCol \Leftrightarrow (G\setminus All(G),k-|All(G)|)\in DualCol.$
\item[(b).] If $All(G)=\emptyset$ but $\overline{G}$ has a matching of size $k$, $x_{1},y_{1},\ldots, x_{k},y_{k}$, with $x_{i}$ being matched to $y_i$ for $i=1,\ldots, k$, then  $(G,k)\in DualCol$ (so reduce it to an arbitrary positive instance). 
\item[(c).] Assume that rules (a),(b) do not apply. Let $(C,H,R)$ be a crown decomposition of the graph $\overline{G}$. Reduce $(G,k)$ to $(G^{\prime},k^{\prime})$, by deleting $H\cup C$ from $G$ and $k^{\prime}=k-|H|$. 
\end{itemize} 

Without loss of generality, we will only apply rule (c). to crown decompositions where $|H|\neq \emptyset$ and all nodes in $C$ are matched to some node in $H$. This is possible for the following reason: if $|H|\neq \emptyset$ and the original crown decomposition had other vertices in $C$, just move them to $R$. If, on the other hand $|H|=\emptyset$ then all vertices in $C$ would be connected to all vertices in $C\cup R$, hence to all vertices of $G$. But this cannot happen, since the case $All(G)\neq \emptyset$ is covered by the first data reduction rule. 

The Crown Decomposition Lemma (Lemma 4.5 of  \cite{fomin2019kernelization}) makes sure that at least one of 
reduction rules (a),(b),(c) applies to every graph with more than $3k-2$ vertices.

The safety of reduction rule (c) can be informally justified as follows: since vertices in a crown decomposition of $\overline{G}$ are matched in a matching $m$, vertices $v\in H$ and $m(v)$ is $C$ are not connected in $G$, hence they can be colored with the same color. At the same time, $m(v)$ is connected in $G$ to all the vertices of $G^{\prime}$, hence must assume  a color different from all the colors of vertices of $G^{\prime}$. Also $m(v_1)$ and $m(v_2)$ are connected, so must assume distinct colors. In conclusion, vertices of $C$ must use $|C|$ different colors, and $G$ is $n-k$ colorable if and only if $G^{\prime}$ is $n-k-|C|$ colorable. But 
$|G^{\prime}|=n-2|C|$, so $G$ is $n-k$ colorable if and only if $G^{\prime}$ is $|G^{\prime}|- (k-|C|)$ colorable. \end{proof}

\subsection{Application: Proof Complexity of Schrijver's Theorem}

In this section we deal with the proof complexity of a stronger version of the Kneser-Lov\'asz theorem known as Schrijver's Theorem \cite{schrijver1978vertex}. This  is a statement about the chromatic number of the so-called \emph{stable Kneser graph} $SKn_{n,k}$,  defined as follows: 

\begin{definition} Call a set $A\subseteq {{n}\choose {k}}$ \emph{stable} if $A$ does not contain two elements that are consecutive (we also consider $n$ and $1$ as consecutive). Denote the set of stable sets by ${{n}\choose {k}}_{st}$. The \emph{stable Kneser graph} $SKn_{n,k}$ is the subgraph of $Kn_{n,k}$ induced by the set ${{n}\choose {k}}_{st}$. 
\end{definition} 

Schrijver's theorem asserts that the chromatic number of the stable Kneser graph $SKn_{n,k}$ is $n-k+2$. We are, of course, interested mainly in the harder part of  this result, the lower bound $\chi(SKn_{n,k})>n-2k+1$. Since $SKn_{n,k}$ is the subset of the Kneser graph (see e.g. Figure~\ref{fig:kneser}, where the central star is the stable Kneser graph $SKn_{5,2}$), this strengthens the (harder part of the) Kneser-Lov\'asz theorem. 
The propositional translation of Schrijver's theorem is immediate, and the resulting unsatisfiable formulas, that we will denote by Schrijver$_{n,k}$ are subformulas of formulas $Kneser_{n,k}$. From Theorem~\ref{theorem-dualcol} we infer the following:

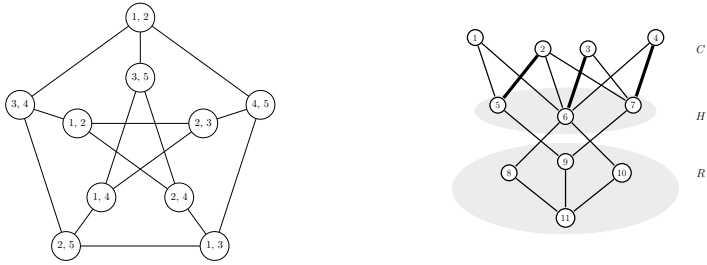
\begin{figure}[h]
\begin{center}
\begin{minipage}{.45\textwidth}
\begin{center} 
\scalebox{0.4}{

\begin{tikzpicture}[every node/.style={draw,shape=circle,text=black}]
    \foreach \i / \ll [count=\ii from 0] in {90/{1, 2}, 162/{3, 4}, 234/ {2, 5}, 306/ {1, 3}, 378/ {4, 5}}
		\path (\i:42mm) node (p\ii) {\ll};
		    \foreach \i / \ll [count=\ii from 5] in {90/{3, 5}, 162/{1, 2}, 234/ {1, 4}, 306/ {2, 4}, 378/ {2, 3}}
		\path (\i:22mm) node (p\ii) {\ll};
		
	\draw (p0) -- (p1);
	\draw (p1) -- (p2);
	\draw (p2) -- (p3);
	\draw (p3) -- (p4);
	\draw (p4) -- (p0);
	\draw (p0) -- (p5);
	\draw (p1) -- (p6);
	\draw (p2) -- (p7);
	\draw (p3) -- (p8);
	\draw (p4) -- (p9);
	\draw (p5) -- (p7);
	\draw (p5) -- (p8);
	\draw (p6) -- (p8);
	\draw (p6) -- (p9);
	\draw (p7) -- (p9);
	
\end{tikzpicture}}%
\end{center} 
\end{minipage} 
\begin{minipage}{.50\textwidth} 
\begin{center} 
\scalebox{0.3}{

\begin{tikzpicture}[shorten >=1pt, auto, node distance=1cm, ultra thick]
    \tikzstyle{node_ordinar} = [circle, draw=black,fill=white!, minimum size=0.7cm]
    \tikzstyle{node_label} = [circle, draw=white,fill=white!, minimum size=0.7cm]
 
    \tikzstyle{edge_style} = [draw=black, line width=2, ultra thick]
    \tikzstyle{edge_match} = [draw=black, line width=4]

    \node[node_ordinar] (v1) at (-4,4) {\large 1};
    \node[node_ordinar] (v2) at (-1,3.5) {\large 2};
    \node[node_ordinar] (v3) at (1,3.5) {\large 3};
    \node[node_ordinar] (v4) at (4,4) {\large 4};

	\draw [gray!15,fill=gray!15] (0,0.75) ellipse (4cm and 1cm);
    \node[node_ordinar] (v5) at (-3,1) {\large 5};
    \node[node_ordinar] (v6) at (0,0.5) {\large 6};
    \node[node_ordinar] (v7) at (3,1) {\large 7};

	\draw [gray!15,fill=gray!15] (0,-2.7) ellipse (5cm and 2cm);
    \node[node_ordinar] (v8) at (-2.5,-2) {\large 8};
    \node[node_ordinar] (v9) at (0,-1.5) {\large 9};
    \node[node_ordinar] (v10) at (2.5,-2) {\large 10};
    \node[node_ordinar] (v11) at (0,-4) {\large 11};

    \node[node_label] (v12) at (6,3.5) {\Large $C$};
    \node[node_label] (v13) at (6,0.5) {\Large $H$};
    \node[node_label] (v14) at (6,-2) {\Large $R$};
    
    \draw[edge_style]  (v1) edge (v5);
    \draw[edge_style]  (v1) edge (v6);
    \draw[edge_match]  (v2) edge (v5);
    \draw[edge_style]  (v2) edge (v6);
    \draw[edge_style]  (v2) edge (v7);
    \draw[edge_match]  (v3) edge (v6);
    \draw[edge_style]  (v3) edge (v7);
    \draw[edge_style]  (v4) edge (v6);
    \draw[edge_match]  (v4) edge (v7);
    \draw[edge_style]  (v5) edge (v9);
    \draw[edge_style]  (v6) edge (v8);
    \draw[edge_style]  (v6) edge (v10);
    \draw[edge_style]  (v7) edge (v9);
    \draw[edge_style]  (v8) edge (v11);
    \draw[edge_style]  (v9) edge (v11);
    \draw[edge_style]  (v10) edge (v11);

 \end{tikzpicture}}
\end{center} 
\end{minipage}
\end{center} 
\caption{(a). The Kneser graph $Kn_{5,2}$.  (b). A crown decomposition of a graph.}
\label{fig:kneser} 
\end{figure}

\begin{corollary} 
For every fixed $k$, formulas $Schrijver_{n,k}$ have polynomial size Frege proofs. 
\end{corollary} 

An alternative, direct quasipolynomial upper bound for Frege proofs of formulas $Schrijver_{n,k}$,  applying our metatheorem to a different kernelization, is given (Theorem~\ref{thm-schrijver}) in the Appendix. 
 
\subsection{Buss Meets Buss: the Proof Complexity of  Vertex Cover}

In this subsection we study the proof complexity of Vertex Cover, the "drosophila of parameterized complexity" \cite{fomin2021kernelization}. We apply our result to a variation of the standard kernelization of VC (called in \cite{fomin2021kernelization} \emph{the Buss reduction}, hence the title of this subsection) to prove: 

\begin{theorem} Instances $(G,k)$ of VC have a kernelization with a data reduction chain of length $O(k)$ to a kernel with at most $k^2$ vertices. The soundness of each step in this data reduction can be witnessed by Frege proofs of size polynomial in $|\Phi(G,k)|$. Hence, for every fixed $k$ negative instances $\Phi(G,k)$ of VC have Frege proofs of size polynomial in $|\Phi(G,k)|$.  
\label{theorem-vc}
\end{theorem} 

\begin{proof} 
 Informally, we will use the following two data reduction rules:  
\begin{description} 
\item[(a).] if $G$ has a vertex $v$ of degree larger than $k$ then G has a VC of size $\leq k$ if and only if $G\setminus \{v\}$ has a VC of size $\leq k-1$. Indeed, $v$ must be part of any VC of $G$ of size $\leq k$. 
\item[(b).] if $Isolated(G)$ denotes the set vertices $v$ in $G$ that are isolated then $G$ has a VC of size $\leq k$ iff $G\setminus Isolated(G)$ has a VC of size $\leq k$. 
\end{description} 
The kernel of these two reduction rules, the set of instances $(G,k)$ of $VC$ such that none of the two rules applies is composed of graphs of at most $k^2$ vertices only \cite{fomin2019kernelization}.
See the Appendix for an encoding of the soundness of these rules by polynomial-size Frege proofs. 
\end{proof} 

\subsection{Proof Complexity of Edge Clique Cover} 

In this section we study the proof complexity of the following problem: 

\begin{definition}[Edge Clique Cover] 
Given graph $G$ and integer $k$, to decide is whether one can find sets of vertices $V_{1},V_{2},\ldots, V_{k}\subseteq V$ s.t. each $V_i$ induces a clique, and for every edge $e=(v,w)\in E$ there exists $1\leq i \leq k$ s.t. $v,w\in V_{i}$ ("each edge is covered by some clique"). We represent instance $(G,k)$ of Edge Clique Cover by propositional formula $\Phi_{G,k}$ as follows: 
\begin{description} 
\item[-] For every pair of distinct vertices $v,w\in V$ define a variable $Y_{v,w}$. For every edge $(v,w)\in E(G)$ add unit clause $Y_{v,w}$. For $(v,w)\not \in E(G)$ add unit clause $\overline{Y_{v,w}}$. 
\item[-] For $v\in V$ and $1\leq i\leq k$  define boolean variable $X_{v,i}=TRUE$ iff $v\in V_{i}$. 
\item[-] For $v,w\in V$ and $1\leq i \leq k$ add $\overline{X_{v,i}}\vee \overline{X_{w,i}}\vee Y_{v,w}$ ("if $v,w\in V_{i}$ then $vw\in E(G)$") and $\overline{Y_{v,w}}\vee (\vee_{j=1}^{k} (X_{v,j}\wedge X_{w,j}))$. Of course, as written above the latter formula is not CNF, but it can be converted easily by expanding the last disjunction. 
\end{description} 
\end{definition} 

The following is our result for the Edge Clique Cover problem. The main technical novelty is reducing the length of the data reduction chain (compared to the usual kernelization) from linear to logarithmic, so that we can get quasipolynomial-size Frege proofs: 

\begin{theorem} 
There exists a kernelization that reduces  instances $(G,k)$ of problem EDGE CLIQUE COVER with graph $G$ having $n$ vertices to a kernel with at most $2^k$ nodes. The length of the data reduction chain is $O(\log_{1+\frac{1}{2^{k}-1}}(n))$. The soundness of each reduction step can be witnessed by polynomial size Frege proofs. Consequently, for fixed $k$, negative instances $(G,k)$ of EDGE CLIQUE COVER have extended Frege proofs of polynomial size and Frege proofs of quasipolynomial size in $|\Phi_{G,k}|$. 
\label{thm-edgecover}
\end{theorem}

\begin{proof} 

We use the following data reduction rules: 
\begin{itemize} 
\item[(a).] If $| Isolated(G) | \geq \frac{n}{2^{k}}$ then reduce $(G,k)$ to $(G\setminus Isolated(G),k)$. 
\item[(b).] If there exists a set $S\subseteq V$, $|S|\geq  \frac{n}{2^{k}}$ such that 
vertices in $S$ induce a clique in $G$, and for all $v,w\in S$ we have N[v]=N[w], where $N[v]$ stands for \emph{the closed neighborhood} of $v$, then reduce $G$ to $(G^{\prime},k^{\prime})$, where $G^{\prime}$ is the graph obtained by identifying vertices $v,w$, and $k^{\prime}=k$ whenever $N[v]=N[w]\neq \emptyset$, $k^{\prime}=k-1$, otherwise. 
\end{itemize} 

\begin{lemma} 
Rules (a). (b). are safe. Also, for every graph $G$ with $n>2^{k}$ vertices one of rules (a). (b). applies. \label{safe-ecc}
\end{lemma}
%\begin{proof} 
%See Appendix.

%\end{proof}  

%The fact that kernel size is $2^k$ is easy to see (see  \cite{fomin2019kernelization} for a formal proof). 

\end{proof}

\subsection{Proof Complexity of the Hitting set problem}

In the $d$-Hitting Set problem we are given an universe $U$ and a family $\mathcal{A}$ of subsets of $U$, all of cardinality at most $d$, as well as an integer $k$. To decide is whether there exists a set $H\subseteq U$ containing at most $k$ elements, such that $H$ intersects every $P\in \mathcal{A}$. 

A formalization of the $d$-Hitting set problem as an instance of SAT is obtained as follows: 

\begin{example} 
Let $P=(U,\mathcal{A},k)$ be an instance of $d$-Hitting set. Define formula $\Phi_{P}$ by: 
\begin{description} 
%\item For all $A\in \mathcal{A}$, add unit clause $Y_{A}$ to the formula. For all $A\subseteq \{1,2,\ldots, n\}$, $|A|=d$, $A\not \in \mathcal{A}$, add unit clause $\overline{Y_{A}}$ to the formula. 
%\item For all $A\subset \{1,2,\ldots, n\}$ \textbf{not} in the instance, $|A|=d$, add unit clause $\overline{Y_{A}}$ (likely not needed). 
\item[-] For $i\in U$, $j=1,\ldots, k$ add variable $X_{i,j}$, TRUE iff $i$ is the $j$'th chosen element. 
\item[-] For $i\neq i^{\prime} \in U$, $1\leq j\leq k$ add clauses $\vee_{i\in U} X_{i,j}$ ("some $i$ is the $j$'th chosen element") and $\overline{X_{i,j}}\OR \overline{X_{i^{\prime},j}}$ ("at most one $i$ can be the $j$'th chosen element"). 
\item[-] For $A\in \mathcal{A}$ add
$
%\overline{Y_{A}}\OR 
(\vee_{i\in A}(\vee_{j=1,\ldots k} X_{i,j}))$ ("some element of $A$ is among the $k$ chosen elements")

\end{description} 
\end{example} 

Our result, which only guarantees  polynomial size \emph{extended} Frege proofs, is: 

\begin{theorem} There exists a kernelization mapping instances $(U,\mathcal{A},k)$ of  $d$-HittingSet with $|U|=n$ elements to a kernel with at most $d\cdot d! \cdot k^d$ sets (hence at most $d^2\cdot d! \cdot k^d$ elements). The data reductions chains in this kernelization have length $O(n^d/k)$, and their soundness can be witnessed by polynomial-size Frege proofs. Hence for every fixed $k,d$, unsatisfiable instances $\Phi_{(U,\mathcal{A},k)}$ of $d$-HittingSet have extended Frege proofs of size $O(poly(|\Phi_{(U,\mathcal{A},k)}|))$. 
\label{theorem-hittingset}
\end{theorem}

\begin{proof} 
We employ the standard kernelization of $d$-Hitting set based on sunflowers: 

\begin{definition} 
A \emph{sunflower with $k$-petals and core $Y$} is a colection of sets $S_{1}, \ldots, S_{k}$, all different from $Y$, such that for $1\leq i<j\leq k$, $S_i \cap S_j = Y$. 
\end{definition} 

We are going to propositionally encode the following informally stated data reduction rule:  let $(U, \mathcal{A}, k)$ be an instance of the $d$-Hitting set such that $\mathcal{A}$ contains a sunflower $\mathcal{S}=\{S_{1},S_{2},\ldots, S_{k+1}\}$ of cardinality $k+1$ with core $Y$. We reduce $(U,\mathcal{A},k)$ to the instance $(U^{\prime}, \mathcal{A}^{\prime}, k)$, where 
$\mathcal{A}^{\prime}=(\mathcal{A}\setminus \mathcal{S}) \cup \{Y\}$  and $U^{\prime}=\cup_{X\in \mathcal{A}^{\prime}} X$.  Indeed, consider a hitting set $H$ for $(U,\mathcal{A},k)$. By definition, $H$ meets every element of $\mathcal{A}\setminus \mathcal{S}$. If $H$ did not meet $Y$ then it would have to meet each of the $k+1$ disjoint petals $S_{j}\setminus Y$. Hence $|H|\leq k$ iff $H$ meets $Y$. To simulate this argument propositionally, see the Appendix. 
\end{proof}

\section{Proof Complexity of principles in Computational Social Choice} 
\label{sec:choice} 

A great number of applications come from the theory of Social Choice: motivated by pioneering work of \cite{tang2009computer}, a significant amount of research in Artificial Intelligence has investigated the provability of such results in logical settings (see \cite{geist2017computer} for a recent survey). We show that the most interesting of these results (Arrow's theorem and the Gibbard-Satterthwaite theorem) have proof complexity counterparts: the unsatisfiability of formulas encoding them can be certified by Frege proofs of subexponential length. 
\iffalse
A new feature of these two problems is that they are in fact "parameterless": although both results are parameterized by two integers ($n$, the number of agents and $m$ the number of choices), the dependency in none of these parameters is exponential. 

\subsection{Proof Complexity of Arrow's Theorem} 
\fi 
A first example of application is Arrow's Theorem. The formulas encoding the nonexistence of a social welfare function satisfying the conditions of Arrow's theorem are rather large. Nevertheless, such an encoding exists, and was used explicitly in \cite{tang2009computer}  to give a computer-assisted proof of  Arrow's theorem\footnote{For encodings of Arrow's Theorem in more powerful logical frameworks %, that fall, however, outside the scope of our work 
see \cite{grandi2009first,cina2015syntactic}}: 
\begin{definition} 
Consider an instance with $n$ agents and $m$ objects to rank. There are $(m!)^{n}$ possible profiles for the complete rankings of the $m$ objects, and $m!$ possible aggregate orderings of the $m$ objects. 
Formula $Arrow_{m,n}$ (unsatisfiable for $m,n\geq 3$) has $(m!)^{n+1}$ variables $X_{R,\pi}$, one for each possible pair $(R,\pi)$ consisting of ranking profile $R$, and an aggregate ordering $\pi \in S_m$. The constraints are the following: 
\begin{description} 
\item[-] For every $R\in \mathcal{R}$ and $\pi_{1}\neq \pi_{2}\in S_m$ add clauses $\bigvee_{\pi \in S_m} X_{R,\pi }$
("every profile is aggregated to some ordering") and $\overline{X_{R,\pi_1}}\vee \overline{X_{R,\pi_2}}$ ("no profile is aggregated to more than one ordering")
\item[-] For $i=1,\ldots, n$ we add to $Arrow_{m,n}$ clauses $\bigvee_{R\in \mathcal{R}} \overline{X_{R,R_{i}}}$. These forbid aggregations that always output the ordering given by the $i$'th agent, i.e. dictatorial rank aggregations. 
\item[-] For every two objects $a,b$ let $S^{m}_{a,b}$ be the set of orderings $\pi$ where for all $i=1,\ldots n$, $\pi^{-1}(a)<\pi^{-1}(b)$ (i.e. $a$ is preferred to $b$ in ordering $\pi$). Let $\mathcal{R}_{a,b}$ be the set of profiles such that for every $i=1,\ldots, n$, $R_{i}\in S^{m}_{a,b}$ (i.e. all agents prefer $a$ to $b$). For every $R\in \mathcal{R}_{a,b}$ add to $Arrow_{m,n}$ clauses $\bigvee_{\pi\in S^{m}_{a,b}} X_{R,\pi}$. 
These constraints encode unanimity (if all agents prefer object $a$ to $b$ then $a$ is preferred to $b$ in the aggregated ranking). 
\item[-] For all profiles $R,R^{\prime}\in \mathcal{R}$ and objects $a,b$ such that all players rank $a,b$ in the same way in both $R,R^{\prime}$ and all pairs $\pi_{1},\pi_{2}\in S^{m}$ that rank $a,b$ in a different way (i.e. $\pi_{1}^{-1}(a)<\pi_{1}^{-1}(b)$ but $\pi_{2}^{-1}(a)>\pi_{2}^{-1}(b)$ or viceversa) we add to $Arrow_{m,n}$ clauses $\overline{X_{R,\pi_1}}\vee \overline{X_{R^{\prime},\pi_2}}$. These encode independence of irrelevant alternatives (if $R,R^{\prime}$ coincide with respect to the relative ordering of $a,b$ then their aggregate orderings also rank $a,b$ in the same way). 

\end{description} 
\label{def-arrow}
\end{definition} 
Results in \cite{tang2009computer} yield a kernelization for $Arrow_{m,n}$ with a reduction chain of length $O(m+n)$. We improve them by providing a kernelization with reduction chains whose length only depends on $n$, implying the existence of polynomial size Frege proofs for constant values of $n$:

\begin{theorem} 
Formulas $Arrow_{m,n}$ have a kernelization with data reduction chains of length $\leq C(n+1)$, with constant  $C$  independent from $m,n$, whose safety is witnessed by polynomial-size Frege proofs. 
Hence (a) formulas $Arrow_{m,n}$ have Frege proofs of size quasipolynomial in $|Arrow_{m,n}|)$. (b). For every fixed $n\geq 3$ there exists a polynomial $p_{n}(\cdot)$ such that for all $m,n\geq 3$ formulas $Arrow_{m,n}$ have Frege proofs of size at most $p_{n}(|Arrow_{m,n}|)$. \label{theorem-arrow} 
\end{theorem} 
%\begin{proof} 
%We note that %$m+n=O(\log(|Arrow_{m,n}|))$, %since this formula has more than %$(m!)^{n}$ variables. 
%\end{proof} 

\begin{proof} 
The kernelization has two data reduction rules, described informally as follows: 
\begin{description} 
\item[(a).] If $n\geq 2, m\geq 6$ and $W:[S_{m}]^{n}\rightarrow [S_{m}]$ is a function that is non-dictatorial, IIA and unanimous then there exists an $T\subseteq [m] $, $|T|=m-5$ such that $W_{-T}:[S_{[m]_{-T}}]^{n}\rightarrow [S_{[m]_{-T}}]$ has the same properties. In other words, one can reduce in one step the set of alternatives from $[m]$ (which has $m$ elements) to $[m]_{-T}$ (which has $5$). 
\item[(b).] See \cite{tang2009computer}: If $n,m\geq 3$ and $W:[S_{m}]^{n}\rightarrow [S_{m}]$ is non-dictatorial, IIA and unanimous then one of functions $W_{1,2},W_{1,3},W_{2,3}:[S_{m}]^{n-1}\rightarrow [S_{m}]$ defined by $W_{i,j}(R_1, R_{2},\ldots, \widehat{R_{i}},\ldots, R_{n})$ $=  W(R^{\prime}_{1}, \ldots,R^{\prime}_{n})$ is non-dictatorial, IIA and unanimous.  Here $R^{\prime}_{i}=R_{j}$, $R^{\prime}(k)=R_{k}, k\neq i$. In other words, one can reduce in one step the number of agents by one. 
\end{description} 

%We now outline the (mathematical) proof of the safety of these reduction rules. 
\begin{lemma} If $W$ is unanimous, IIA and non-dictatorial then for every $B\subseteq [m]$, function $W_{-B}$ is unanimous and IIA. 
\label{arrow-red}
\end{lemma} 
%\begin{proof} 
%See Appendix. 
%\end{proof} 

\begin{lemma} 
Reduction (a). is safe. 
\label{safe-a}
\end{lemma} 
\begin{proof} 
Consider an arbitrary set $T\subseteq [m]$ of cardinality $m-6$, e.g. $T=\{7,\ldots, m\}$. Let $x\not \in T$, e.g. $x=6$ and $U=T\cup \{x\}$. If $W_{-U}$ is non-dictatorial we are done. Otherwise, assume w.l.o.g. that agent $1$ is a dictator for $W_{-U}$. Since $1$ is not a dictator for $W$, there must exist indices $c\neq d\in [m] $ and preference profiles $<_1,\ldots, <_{n}$ on $[m]$ such that $c<_{1}d$ but $d<_{W(<_{1},\ldots,<_{n})}c$. %Clearly, $c,d\not \in T$, since $<_{1}$ and $<_{W(<_{1},\ldots,<_{n})}$ coincide on $T$. 
Let $y,a,b\not \in T$, different from $x,c,d$, and let $V\subseteq [m]$, $|V|=m-5$, $a,b,c,d,x\not\in V$, $y\in V$. Such a $V$ exists, since $m\geq 6$. Clearly $V\neq U$, since $y\in V\setminus U$. We claim that  function $W_{-V}$ is not dictatorial. 

First note that 1 cannot be a dictator for $W_{-V}$. Indeed, consider $<_{i,-V}$ the restriction of $<_{i}$ to $[m]_{-V}$. We have $c<_{1,-V}d$ but $d<_{W_{-V}(<_{1,-V},\ldots,<_{n,-V})} c$. The first relation holds because $c<_{1}d$ and $c,d\not\in V$. The second relation holds because to compare $c,d$ according to $W_{-V}(<_{1,-V},\ldots,<_{n,-V})$ we apply function $W$ on $(<^{+V}_{1,-V},\ldots,<^{+V}_{n,-V})$. But since $c,d\not \in V$, $<^{+V}_{i,-V}$ coincides with $<_{i}$ with respect to the ordering of $c,d$ for $i=1,\ldots, n$. Invoking the IIA property of $W$ for tuples $(<_1,\ldots, <_{n})$ and $(<^{+V}_{1,-V},\ldots,<^{+V}_{n,-V})$ justifies the second relation.  

Assume now that another agent, say 2, were a dictator for $W_{-V}$. Let $<_{1}, <_{2}$ be preference profiles on $[m]$ s.t. $a<_{1} b$ but $b<_{2} a$, and $<_{3},\ldots, <_{n}$ be arbitrary preference profiles. First, invoking the first relation,  the fact that $W_{-U}$ is computed using $W$, and that 1 is a dictator for $W_{-U}$ we get that $a<_{W_{-U}(<_{1,-U}, <_{2,-U},\ldots, <_{n,-U})} b$. Invoking the IIA property of $W$ on tuples $(<_{1,-U}, <_{2,-U},\ldots, <_{n,-U})$ and $(<_{1},<_{2},\ldots, <_{n})$ we get that $a<_{W(<_{1}, <_{2}, \ldots, <_{n})} b$. Using a similar reasoning for the function $W_{-V}$ we get that $b<_{W(<_{1}, <_{2},\ldots, <_{n})} a$, a contradiction. 
\end{proof} 
The safety of reduction (b) was (mathematically) proved in \cite{tang2009computer}. In the Appendix we outline how to simulate these mathematical arguments using polynomial-size Frege proofs. We further note that $n=O(\log(|Arrow_{m,n}|))$. Invoking our metatheorem yields a proof of point (a) of Theorem~\ref{theorem-arrow}. Point (b) follows by invoking point 2 of the same metatheorem. 

%\begin{proof} 
%See the Appendix. 
%\end{proof} 
\end{proof} 
\iffalse
\subsection{The Proof Complexity of the Gibbard-Satterthwaite theorem} \fi 

As for the Gibbard-Satterthwaite theorem, we use the following formalization: 
\iffalse
\begin{definition} 
Consider an instance with $n$ agents and $m$ objects to rank. There are $(m!)^{n}$ possible profiles for the complete rankings of the $m$ objects. Denote their set by $P_{m,n}$. A \emph{social choice function} is a function $f:P_{m,n}\rightarrow \{1,2,\ldots, m\}$. We denote by $sc(m,n)$ the set of social choice functions. A function $f\in sc(m,n)$ is called: 
\begin{description} 
\item[-] Onto if for every $i\in \{1,2,\ldots, m\}$ there exists $R\in P_{m,n}$ such that $f(R)=i$. 
\item[-] Dictatorial if there exists $j\in \{1,2,\ldots, n\}$ such that for every $R\in P_{m,n}$, $f(R)=first(R_{j})$. 
\item[-] manipulable by player $i$ at profile $R$ if there exists $\pi\in S_{m}$ such that replacing $R_{i}$ by $\pi$ in $R$ yields a profile $Q$ such that $R_{i}^{-1}(f(Q))<R_{i}^{-1}(f(R))$ (in other words $i$ prefers the outcome of $f(Q)$ to the outcome of $f(R)$).
\item[-] strategyproof if $f$ is not manipulable by any player at any profile.  
\end{description} 
\end{definition} 
\fi 
\begin{definition} 
Consider an instance with $n$ agents and $m$ objects. There are $(m!)^{n}$ possible profiles for the complete rankings of the $m$ objects, and $m$ possible outcomes. 
Formula $GS_{m,n}$ has $(m!)^{n}\times m$ variables $X_{R,o}$, one for each possible pair consisting of a strategy profile $R$ and a value $o\in [m]$,  the value of the SCF on profile $R$. The constraints are the following: 
\begin{description} 
\item[-] For $R\in \mathcal{R}$ add clauses $\bigvee_{o \in [m]} X_{R,o}$
("every joint profile is aggregated to some object") and  $\overline{X_{R,o_1}}\vee \overline{X_{R,o_2}}$ ("no joint profile is aggregated to more than one object").
\item[-] For $i=1,\ldots, n$ we add to $GS_{m,n}$ clauses $\bigvee_{R\in \mathcal{R}} \overline{X_{R,top(R_{i})}}$. They forbid social choice functions that always output the top preference of the $i$'th agent, i.e. dictatorial aggregations. 
\item[-] For $i=1,\ldots, n$ add $\bigvee_{R\in \mathcal{R}} X_{R,o}$. This eliminates social choice functions that are not onto. 
\item[-] Add, for every pair $(R,o)$,  player $1\leq i\leq m$ and  $\pi\in S_{m}$, 
$$
 \overline{X_{R,o}} \OR (\bigvee_{o^{\prime}\in pr(i,o,R)} X_{s(i,R,\pi),o^{\prime}}).$$ 
These clauses state that  the social choice function is strategyproof. 
\end{description} 
%By the Gibbard-Satterthwaite theorem, for $m\geq 3$, $n\geq 2$ formulas $GS_{m,n}$ are unsatisfiable. 
\end{definition} 

\begin{theorem} 
For every fixed $m$ formulas $GS_{n,m}$ have a kernelization of length $O(n)$ whose soundness has polynomial time Frege proofs. Hence, formulas $GS_{n,m}$ expressing the Gibbard-Satterthwaite theorem have (a). Frege proofs of size quasipolynomial in $|GS_{n,m}|$, (b). for every fixed $n$, Frege proofs of size polynomial in $|GS_{n,m}|$. 
\label{theorem-gs}
\end{theorem} 

\begin{proof} 
The argument is very similar to that of Arrow's theorem: We use the following two data reductions:  (a). for $n\geq 2, m\geq 4$, 
if $W:[m]^{n}\rightarrow [m]$ is a social choice function that is onto, non-dictatorial and strategy-proof then there exists $T\subseteq [m]$, $|T|=m-3$ such that function $W_{-T}:([m]_{-T})^{n}\rightarrow [m]_{-T}$ is onto, non-dictatorial and strategyproof, and (b). \cite{tang2008computer} if $W:[m]^{n}\rightarrow [m]$ is a social choice function that is onto, non-dictatorial and strategy-proof then one of functions $W_{1,2},W_{1,3},W_{2,3}:[S_{m}]^{n-1}\rightarrow [m]$ defined by $W_{a,b}(R_1, R_{2},\ldots, \widehat{R_{a}},\ldots, R_{n})=  W(R^{\prime}_{1}, \ldots,R^{\prime}_{n})$ is non-dictatorial, onto and strategy-proof.  Here $R^{\prime}_{a}=R_{b}$, $R^{\prime}(k)=R_{k}$ for $k\neq a$. 

\begin{lemma} If $W:S_{m}^{n}\rightarrow [m]$ is onto, strategyproof and non-dictatorial then for every $B\subseteq [m]$, function $W_{-B}$ is onto and strategyproof. 
\label{gs-red}
\end{lemma} 

For mathematical proofs of the safety of reduction rule (a) (for (b) see \cite{tang2008computer}) and some details on the propositional simulations see the Appendix. 
\end{proof} 

\section{Conclusions and open problems} 

We believe that the most important contribution of our paper is to show that several techniques for proving kernelization can sometimes be simulated by efficient extended Frege proofs. The proof techniques in this list includes: crown decomposition, the sunflower lemma, ad-hoc methods.  It is an interesting challenge to enlarge the list of methods and problems that have such a simulation or, conversely, show limits of some of these methods.

It is rewarding to note that methods from algebraic topology  can be used to reframe and extend results in both Topological Combinatorics and Computational Social Choice: Arrow's theorem has topological proofs \cite{chichilnisky1980social,chichilnisky1982topological,baryshnikov1997topological}. On the other hand the original results on Kneser's conjecture \cite{lovasz-kneser} have been strengthened using more advanced topological methods, e.g. \cite{babson2007proof} (see \cite{kozlov-book} for a book-length treatment). The results in \cite{tang2009computer,aisenberg2018short} can be interpreted as stating that in both cases one can bypass topological arguments by purely combinatorial arguments (plus computer-assisted verification of finitely many cases). It is an interesting question whether this is still true for the results requiring more sophisticated topological methods as well. 

In Theorem~\ref{theorem-hittingset} we have only obtained polynomial size extended Frege proofs. Similarly, in Theorems~\ref{theorem-arrow},~\ref{theorem-gs} we have only obtained polynomial-size Frege proofs when the number of agents is fixed. We leave open the issue of improving these results. On the other hand, our results have only touched on the most basic topics on the proof complexity of statements in computational social choice. There has been significant progress in this area (see e.g. \cite{geist2017computer}) and we believe that our framework may be applicable to some of this work (e.g. to the Preservation Theorem of \cite{geist2011automated}). It would be interesting to see if this is really the case. 

Finally, results in this paper motivate the question of developing a proof complexity theory for richer frameworks, e.g. \emph{satisfiability modulo theories} \cite{kroening2016decision}: by employing more powerful background theories many logical statements, including those from computational social choice, could be encoded in more compact and natural ways. Methods based on SMT have recently been indeed used for the automated derivation of statements from computational social choice \cite{brandl2018proving}. For some preliminary steps in developing such a theory of proof complexity for SMT see \cite{robere2018proof}. 

\newpage

\iffalse
\newcommand{\etalchar}[1]{$^{#1}$}

\fi 
%\bblname{paper}
%\bibliography{/home/gistrate/Dropbox/texmf/bibtex/bib/bibtheory}
\bibliography{/Users/gistrate/Dropbox/texmf/bibtex/bib/bibtheory}

\newpage
\section*{Appendix}

In this section we provide \textbf{sketches} of the proofs ommitted in the main text. 

First of all, we will need the following classical result: 
\begin{lemma}
The number of vectors $(x_{1},\ldots, x_{k})$ of solutions of equation $x_{1}+\ldots + x_{k}=n$ in nonnegative integers is ${{n+k-1}\choose {k-1}}$. 
\label{one}
\end{lemma}
\begin{proof}
Simplest way to see it: consider numbers $x_1+1$, $x_1+x_2+2, \ldots, x_1+\ldots+x_{k-1}+(k-1)$. This gives a subset of $k-1$ numbers in $1,\ldots, n+k-1.$ Conversely, from any set $y_1,y_2,\ldots, y_k$ we generate solution $(x_1,x_2,\ldots x_k)$ with $x_1=y_1-1,x_2=y_1+y_2-2,\ldots, y_k=n - (y_1+\ldots + y_{k-1}).$
\end{proof}

\section{Proof of Theorem~\ref{thm-main}}

\begin{enumerate} 
\item Given an instance $(x,k)$ of $L$, and data reduction chain $(x,k)=(x_{0},k_{0}), (x_{1},k_{1}),$ $\ldots  (x_{m},k_{m})$, an extended Frege proof for $\Phi(x,k)$ is obtained by concatenating the proofs for statements $\Phi(x_{i-1},k_{i-1})\vdash \Phi(x_{i},k_{i})$ with an extended Frege proof of the kernel instance. There is one complication, though, induced by the fact that we allow at most $R$ cases in the reduction: the reduction chain maps to a \textbf{tree} of propositional proofs, since for each node $\Phi(x_{i},k_{i})$ we have $r^{\prime}_{i}\leq R$ children $\Xi_{i,t}$, all isomorphic to $\Phi(x_{i+1},k_{i+1})$ (but different). The total number of nodes in this tree is at most $\sum_{t=0}^{C(x,k)} R^{t}$. 

 Then the whole chain of reductions from $\Phi(x,k)$ to $\Phi(x_{m},k_{m})$ can be proved to be sound by proofs of length $(\sum_{t=0}^{C(x,k)} R^{t}) \cdot h(|\Phi(x,k)|)$. There are at most $R^{C(x,k)}$ copies of the kernel instance. Each of them  can be proved (in brute force) in size $O(2^{|\Phi(x_{m},k_{m})|})=O(2^{poly(g(k))})$, since $(x_{m},k_{m})\in ker(L)$ and any unsatisfiable formula $\Xi$ with $n$ variables has Frege proofs of size $O(2^{n})$.

The length of the total proof is thus $O(\sum_{t=0}^{C(x,k)} R^{t}\cdot [h(|\Phi(x,k)|)+ 2^{poly(g(k))})]$. We infer the desired result when $C(x,k)=O(poly(|\Phi(x,k)|))$. 

\item We unwind the substitions implicit in the extended Frege proofs. For $R=1$ (i.e. a reduction chain), arguing that the blow-up due to making substitutions is quasipolynomial as long as the chain length is logarithmic is identical to similar arguments made in \cite{buss2015quasipolynomial}, \cite{aisenberg2018short} for other problems, and we omit further details. 

In our case the complication arises since we no longer have a chain but a tree. However, we can upper bound the complexity of Frege proofs by $R^{C(x,k)}$ times the complexity of a single chain (a root-to-leaf path in this tree). As long as $C(x,k)=O(\log (|\Phi(x,k)|))$, the term $R^{C(x,k)}$ has a magnitude polynomial in $|\Phi(x,k)|$. Multiplying this polynomial by the quasipolynomial complexity of each chain still yields a proof of complexity quasipolynomial in $|\Phi(x,k)|$. 
\end{enumerate} 

\begin{observation} 
There is an important uniformity aspect of kernelization that we haven't used in the preceding proof: \emph{the fact that data reductions are specified by polynomial time algorithms.} This issue will be important in applying the above result: often the existence of a data reduction is proved by an algorithm whose soundness (for all instances) would be rather cumbersome to simulate in propositional proofs. This is the case when results involve general techniques for developing kernelizations, such as the Crown Decomposition Lemma or the Sunflower Lemma. As long as we do not insist, however, on \emph{actually generating} the proof, but merely on \emph{proving its existence}, we can get away with proving the soundness of individual instances. That is, if we can prove the soundness of \emph{an individual application of a  propositional reduction rule},  $\Phi(x,k)\vdash \Phi(x^{\prime},k^{\prime})$, taking for granted the existence/definition of $(x^{\prime},k^{\prime})$, we can prove the existence of efficient proofs, \textbf{without actually having to generate a propositional proof of the soundness of the reduction techniques.} 
\end{observation}

\section{A direct quasipolynomial Upper Bound on the proof complexity of Schrijver's Theorem}

 \begin{theorem} For every fixed $0<\beta<1$ and $k$, Schrijver's Theorem has a kernelization of length $O(\log_{1+\frac{\beta}{k-\beta}}(n))$ whose soundness can be established by polynomial size Frege proofs. Hence, for every fixed $k$ the class of formulas $(Sch_{n,k})_{n\geq 1}$ has extended Frege proofs of size polynomial in $|Sch_{n,k}|$, as well as Frege proofs of size quasipolynomial in $|Sch_{n,k}|$. \label{thm-schrijver} 
 \end{theorem} 
 
 \begin{proof}

 The following is an easy result, for which we haven't found a formal reference: 
\begin{lemma}
The number of vertices of the stable Kneser graph $SKn_{n,k}$ (i.e. the cardinal of the set ${{n}\choose {k}}_{st}$) is 
$
{{n-k+1}\choose {k}}+{{n-k}\choose {k-1}}$. 
Also, for every $x\in 1,\ldots, n$ the number of sets in ${{n}\choose {k}}_{st}$ containing $x$ is $ {{n-k}\choose {k-1}}$.
\label{cardinal-stable-kneser}
\end{lemma}

\begin{proof} 
%The result is clear for $k=1$. 

To encode a stable set $\{x_1,x_{2},\ldots, x_{k}\}\subseteq [n]$ we have to give: \begin{description}
\item[-] $a_1=x_1$, the distance between point $0(=n)$ and $a_1.$
\item[-] $a_i=x_{i}-x_{i-1}$, that is the distance between $x_{i-1}$ and $x_i.$
\item[-] $a_{k+1}$, defined as the number of positions between point $x_k$ and point $0(=n)$ on the circle, going clockwise. 
\end{description}

We divide the counting of stable sets in two cases: 
\vspace{5mm}

\textbf{Case 1: $a_1=0.$} Then 
$a_{2}, \ldots, a_{k+1}\geq 2$, $a_{2}+\ldots + a_{k+1}=n$. Denoting $b_i=a_i-2$ we get $b_i\geq 0,$ $b_2+\ldots + b_{k+1}=n-2k.$

By Lemma~\ref{one}, the  number of such tuples is ${{n-k-1}\choose {k-1}}$. 

\vspace{5mm}

\textbf{Case 2: $a_1\geq 1.$} Then $a_{2}, \ldots, a_{k}\geq 2$, $a_{k+1}\geq 0$,  $a_{1}+\ldots + a_{k+1}=n$. 

Define $b_1=a_1-1$, $b_2=a_2-2, \ldots, b_k=a_k-2, b_{k+1}=a_{k+1}.$ Thus 
$b_i\geq 0$ and 
$b_{1}+\ldots + b_{k+1}=n-1-2(k-1)$. By Lemma~\ref{one}, the number of such triplets is ${{n-2k+1+(k+1)-1}\choose {k}}= {{n-k+1}\choose {k}}$

As for the second part, 
Similarly to the previous proof, $a_1=1,a_{2},\ldots, a_{k}\geq 2, a_{k+1}\geq 0$.

Denoting $b_1=a_1-1$, $b_i=a_i-2$ for $i=2,\ldots k$ and $b_{k+1}=a_{k+1}$ we get $b_2+\ldots + b_{k+1}=n-2(k-1)-1= n-2k+1$. The number of solutions is ${{n-2k+1+k-1}\choose {k-1}}$

\end{proof}

Call a family of sets \emph{star-shaped} if all set in the family share a fixed element. 
Talbot \cite{talbot-separated} proved the (first part of the) following result: 

\begin{lemma}
Let $C\subseteq {{n}\choose {k}}_{st}$, be a non-star shaped set. Then $
|C|\leq |\{A\in {{n}\choose {k}}_{st}: 1\in A\}|\mbox{  } (= {{n-k}\choose {k-1}} )$, 
with the second line following from Lemma~\ref{cardinal-stable-kneser}. 
\end{lemma} 

We will, however, simulate propositionally a weaker relative of Talbot's theorem:

\begin{theorem}
Let $C\subseteq {{n}\choose {k}}_{st}$ be a non-star shaped set. Then $
|C|\leq k^2{{n+k-1}\choose {k-2}}$. 
\label{upper-schrijver}
\end{theorem} 
\begin{proof}

Suppose $C$ is not star-shaped and nonempty. Let $S_0=\{a_1,a_2,\ldots, a_k\}$ be some fixed set in $C$. Since $C$
is not star-shaped, there must be sets $S_{1},S_{2},\ldots, $ $S_{k}\in C$ with $a_i\not \in S_i$ for $i=1,\ldots, k.$ To specify an element $S$ of $C$ we first specify some $a_i\in S\cap S_0.$ Such an element exists since $S$ and $S_0$ are in the same color class, hence they must intersect. Similarly, $S$ and $S_i$ must intersect, hence they have a common element $a_j$. $a_i\neq a_j$, since $a_j\in S_{i}$ but $a_i\not \in S_i$. 

We will view numbers from $1$ to $n$ on a circle (to make $1$ and $n$ neighbors). Then $a_i$ and $a_j$ are not neighbors in this representation. We prove that the number of stable sets of size $k$ containing $a_i,a_j$ is maximized when $a_i,a_j$ are at distance two from eachother: 
 
\begin{lemma} Given two non-neighboring elements $a,b$, there are at most ${{n+k-1}\choose {k-2}}$ stable sets that contain both elements $a,b.$ The bound is tight, being realized e.g. for two points $a_0=1,b_0=3$ at distance two on the circle. 
\label{stable-two}
\end{lemma}

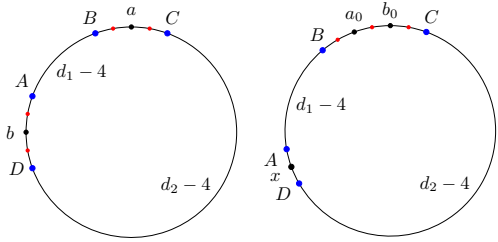
\begin{figure}[h]
\begin{center}
\scalebox{0.7}{

\begin{tikzpicture}
  \draw (-2.5,-.05) node (a) {} 
  arc(0:70:2cm) node (b1) {}
  arc(70:80:2cm) node (r1) {}
  arc(80:90:2cm) node (aa) {}
  arc(90:100:2cm) node (r2) {}
  arc(100:110:2cm) node (b2) {}
  arc(110:125:2cm) node (d1) {}
  arc(125:160:2cm) node (b3) {}
  arc(160:170:2cm) node (r3) {}
  arc(170:180:2cm) node (bb) {}
  arc(180:190:2cm) node (r4) {}
  arc(190:200:2cm) node (b4) {}
  arc(200:330:2cm) node (d2) {}
  arc(330:360:2cm) node (dd) {};
  %arc(180:270:2cm);
  %\draw[fill=black] (a) circle (.02);
  \foreach \i in {1,2,...,4}{
  		\draw[red,fill=red] (r\i) circle (.03);
  		\draw[blue,fill=blue] (b\i) circle (.05);
  }
  \draw[black,fill=black] (aa) circle (.04); 
  \draw[black,fill=black] (bb) circle (.04); 

\draw(b1) ++(0.1,.3cm) node {$C$};
\draw(b2) ++(-0.1,.3cm) node {$B$};
\draw(b3) ++(-0.2,.3cm) node {$A$};
\draw(b4) ++(-0.3,0cm) node {$D$};

  \draw (aa) ++(0,.3cm) node {$a$};
  \draw (bb) ++(-.3cm,0) node {$b$};
  \draw (d1) ++(0.2,-.5cm) node {$d_1-4$};
  \draw (d2) ++(-.7cm, 0) node {$d_2-4$};
\end{tikzpicture}

}
\hspace{-0.2cm}
%Star Graph
\scalebox{0.7}{

\begin{tikzpicture}
  \draw (-2.5,-.05) node (a) {} 
  arc(0:70:2cm) node (b1) {}
  arc(70:80:2cm) node (r1) {}
  arc(80:90:2cm) node (bb) {}
  arc(90:100:2cm) node (r2) {}
  %arc(100:110:2cm) node (r3) {}
  arc(100:110:2cm) node (aa) {}
  arc(110:120:2cm) node (r3) {}
  arc(120:130:2cm) node (b2) {}
  arc(130:155:2cm) node (d1) {} 
  arc(155:190:2cm) node (b3) {} 
  arc(190:200:2cm) node (x) {}
  arc(200:210:2cm) node (b4) {}
  arc(210:330:2cm) node (d2) {}
  arc(330:360:2cm) node (dd) {};

  %\draw[fill=black] (a) circle (.02);
  \foreach \i in {1,2,...,3}{
  		\draw[red,fill=red] (r\i) circle (.03);
  }

  \foreach \i in {1,2,...,4}{
  		\draw[blue,fill=blue] (b\i) circle (.05);
  }
  \draw[black,fill=black] (x) circle (.05);
  \draw[black,fill=black] (aa) circle (.04); 
  \draw[black,fill=black] (bb) circle (.04); 

\draw(b1) ++(0.1,.3cm) node {$C$};
\draw(b2) ++(-0.1,.3cm) node {$B$};
\draw(b3) ++(-0.3,-.2cm) node {$A$};
\draw(b4) ++(-0.3,-.2cm) node {$D$};

  \draw (aa) ++(0,.3cm) node {$a_0$};
  \draw (bb) ++ (0,.3cm) node {$b_0$};
  \draw (x) ++(-.3,-.2cm) node {$x$};
  %\draw (d1) ++(.5cm, -.1cm) node {$d_1-4$};
  %\draw (d2) ++(-.5cm, 0) node {$d_2-4$};
    \draw (d1) ++(0.5,-.35cm) node {$d_1-4$};
  \draw (d2) ++(-.7cm, 0) node {$d_2-4$};
\end{tikzpicture}

}
\end{center}
\caption{The injection in the proof of Lemma~\ref{stable-two}: (a). Initial stable sets. (b). Target stable sets.}
\label{fig:lemma4} 
\end{figure}
\begin{proof} 
We will show the desired maximization result by constructing an injective mapping of stable sets all to the set of stable sets containing $a_0=1,b_0=3$. Indeed, let  $C_{1}$ be an arbitrary stable set containing two fixed integers $a,b$. Let 
$d_1, d_2\geq 2$ be the (left/right) distances between $a,b$ when viewed on a circle. We may assume that $d_1,d_2\geq 3$, otherwise we are counting stable sets in a case isomorphic to those of the image.   Neighbors of $a,b$ are forbidden by the stability constraint from being in the set, hence 

\begin{description}
\item[-] If $d_1= 3$ (or, symmetrically, if $d_2= 3$) then no point on the segment of length three between $a$ and $b$ can be part of the stable set. Thus all the points in the stable set except $a,b$ lie in a circular segment of length $n-7$.  For stable sets containing points $a_0,b_0$, on the other hand, all other points  except $a_0,b_0$ lie in circular segment of length $n-6$, so the inequality is evident. 

\item[-] If $d_1,d_2\geq 4$ then removing $a,b$ and their (forbidden) neighbors creates two circular segments of lengths $d_1-4,d_2-4$, respectively, containing all the points of the stable set, other than $a,b$ (bounded by the blue points in Fig.~\ref{fig:lemma4} (a)). Let $x$ be the unique point at distance $d_1-1$ from $a_0$ in the direction opposite to $b_0$ (Fig.~\ref{fig:lemma4} (b)). If $d_1,d_2\geq 4$ we actually get a bijection between stable sets of size $k$ containing $a,b$ to stable sets of size $k$ containing $a_0,b_0$ but \textbf{not} $x$. The bijection is self-evident: we map bijectively the circle segments  $AB$ ($CD$) in the first figure onto circle segments $AB$ ($CD$) in the second figure,  respectively. 

\end{description} 

All we need now is to count the maximum number of stable sets of size $k$ containing $a_0,b_0$. This is accomplished with the help of the following
\begin{lemma}
In a circular segment having numbers $a,b$ as extremities, the number of stable sets containing exactly $r$ elements is ${{b-a-r+2}\choose {r}}$. 
\label{cnt}
\end{lemma}
\begin{proof}
We can encode a stable set $\{y_1,y_2,\ldots, y_r\}\subseteq \{a,\ldots, b\}$ by a vector $(x_{1},\ldots, x_{r+1})$, where:  
\begin{itemize} 
\item[-] $x_1=y_1-a\geq 0.$ 
\item[-] $x_2=y_2-y_1-2\geq 0.$
\item[-] $x_r=y_r-y_{r-1}-2\geq 0.$
\item[-] $x_{r+1}=b-y_r\geq 0.$
\end{itemize} 
We have $x_1+\ldots + x_{r+1}=b-a-2(r-1).$ The number of vectors of solutions to this equation is ${{b-a-2(r-1)+r}\choose {r}}$. 
\end{proof}

To apply this lemma to our problem, we note that the outer circle segment between points $a_0-2(=n-1)$ and $b_0+2(=5)$ has length $n-5$ and must contain $k-2$ points.

\end{proof}
Using the conclusion of Lemma~\ref{stable-two} proves Theorem~\ref{upper-schrijver}.

\end{proof}

\begin{lemma}
Fix $k>1$ and $0<\beta<1$. Then there exists $N(k,\beta)$ such that for $n>N(k,\beta)$ in any $(n-2k+1)$-coloring of ${{n}\choose {k}}_{st}$ there are at least $\frac{n\beta}{k}$ star-shaped color classes. 
\label{lemma-beta}
\end{lemma}
\begin{proof} 
Suppose $c$ is a $(n-2k+1)$-coloring of ${{n}\choose {k}}_{st}$, and let $\alpha$ be the number of star-shaped classes of $c.$

Any star-shaped class has at most
${{n-k-1}\choose {k-1}}$ elements: this can be seen by applying Lemma~\ref{cnt} to segment $3,n-1$ which has length $n-3$ elements, and we must choose $k-1$ elements, apart from the center, to form a stable set. 

Any non-star-shaped color class has at most $k^2{{n+k-1}\choose {k-2}}$ elements.  
Let $\alpha$ be the number of star-shaped color classes. Then we upper bound the number of vertices in $SKn_{n,k}$ using the two previously-mentioned upper bounds as follows: 

\begin{align*}
\alpha {{n-k-1}\choose {k-1}} + (n-2k+1-\alpha) k^2{{n+k-1}\choose {k-2}}\geq 
{{n-k+1}\choose {k}}+{{n-k}\choose {k-1}}, \mbox{ or }
\end{align*} 

\begin{align*}
\alpha \geq   \frac{{{n-k+1}\choose {k}}+  {{n-k}\choose {k-1}}- (n-2k+1) k^2{{n+k-1}\choose {k-2}}}{({{n-k-1}\choose {k-1}}- k^2{{n+k-1}\choose {k-2}})}
\end{align*} 

The polynomial (in $n$) in the numerator has degree $k-1$ and leading term $\frac{n^{k-1}}{(k-1)!}(1+o(1))$, whereas the one in the denominator has degree $k$ and leading term $\frac{n^{k}}{k!}(1+o(1))$. Hence, for every $0<\beta<1$ there exists $N(k,\beta)$ so that for $n>N(k,\beta)$ we have $\alpha > \frac{n\beta}{k}$. 
\end{proof} 

Choose now $0<\beta<1$.  We will use the following data reduction rule, guaranteed to work for $n>N(k,\beta)$, the constant in Lemma~\ref{lemma-beta}: 
\begin{equation} 
(SKn_{n,k},n-2k+1) \vdash (SKn_{n-\frac{n\beta}{k}}, n(1-\frac{\beta}{k})-2k+1)
\end{equation} 
That is, we will eliminate in one round $\frac{n\beta}{k}$ star-shaped color classes, and equally many colors. This will ensure a data reduction chain of size $O(log_{1+\frac{\beta}{k-\beta}}(n))$. 

All we need is to show that the soundness of this reduction rule can be witnessed by polynomial size Frege proofs. The formalization is essentially similar to the one given in \cite{aisenberg2018short}  establishing quasipolynomial Frege proofs (and, implicitly, polynomial size extended Frege proofs) for the ordinary Kneser formulas. 

\end{proof}

\section{Proof of Theorem~\ref{theorem-dualcol}}

Rule (b). does not apply to unsatisfiable instances of DualCol. Hence we have to argue about the size of Frege proofs witnessing the soundness of rules (a) and (c), namely: 
Let $\Phi_{v,1}[\overline{Y}]$ be the formula $\wedge_{w\neq v} Y_{v,w}$ (informally, $v\in All(G)$). We need to provide proofs that witness that
\[
\Phi(G,k)\wedge \Phi_{v,1}[\overline{Y}] \vdash \Phi(G\setminus \{v\},k-1),\mbox{ and }
\]
\[
\Phi(G,k)\vdash \Phi(G^{\prime},k^{\prime}). 
\]
For the first implication, define new variables $Z_{w,i}$ via the substitution, for $w\neq v\in V(G)$,  $Z_{w,i^{\prime}}\leftrightarrow X_{w,i}\AND X_{v,l}$, where $i^{\prime}=i$ for $i<l$, $i^{\prime}=i-1$ for $i>l$. 

We start by deriving, by resolving unit literals $Y_{v,w}$ (which are part of the formula), for all $w\neq v\in V(G)$ and $i$, clauses $\overline{X_{v,i}}\OR \overline{X_{w,i}}$. Then we derive, for every $w\neq v\in V(G)$ and $i$, clauses $\overline{X_{v,i}}\OR (\OR_{j\neq i}X_{w,j})$. This is done by resolving $(\OR_{j=1}^{k}X_{w,j})$ and $\overline{X_{v,i}}\OR \overline{X_{w,i}}$. We then derive clauses $\overline{X_{v,i}}\OR (\OR_{i^{\prime}=1}^{k-1} Z_{w,i^{\prime}})$. By resolving all these clauses against $\OR_{i=1}^{k} X_{v,i}$ we derive $(\OR_{i^{\prime}=1}^{k-1} Z_{w,i^{\prime}})$. Similar tricks allow deriving clauses $\overline{Z_{w,i}}\OR \overline{Z_{w,j}}$ and  $\overline{Y_{v,w}}\OR \overline{Z_{v,i}}\OR \overline{Z_{w,i}}$ from the corresponding clauses in the $X$ variables. 

As for the second reduction rule, intuitively we want to encode the fact that if a vertex $v\in G^{\prime}$ is colored with color $i$ in $G$, then coloring it with color $i-less(v)$, where $less(v)$ is the number of nodes in $H$ colored with a color smaller than $i$, yields a legal coloring of $G^{\prime}$. This is true since all nodes in $H$ must get colors (in $G$) different from all colors in $G^{\prime}$.

For an arbitrary vertex $v\in G^{\prime}$, let $less(v)$ be the number of nodes $w\in C$ (here $C$ refers to the class of the crown decomposition of $\overline{G}$) such that $col(w)<col(v)$. One can compute the binary representation of number $less(v)$ using Frege proofs as follows: we create a boolean variable $T_{v,w}$ which will be true iff $col(w)<col(v)$. One can compute $T_{v,w}$ as %by using the circuits for $col(v)$ and  $col(w)$ and then comparing the bits of the binary representation of the two colors. This yields a formula $Less_{v,w}$. 
\[
T_{v,w}:= \bigvee_{i<j} X_{w,i}\AND X_{v,j}.
\]
Now we simply use the predicate $COUNT(T_{v,w})_{w\in C}$ to compute the binary representation of $less(v)$. Here COUNT is the Buss counting predicate \cite{buss-frege-php}. We will also derive the following formulas: 
\begin{equation} 
\overline{X_{v,i}}\vee \bigvee_{t=1}^{i} [less[v]=t] 
\label{feq}
\end{equation} 
To accomplish that, we use the pigeonhole principle $PHP_{i}^{i+1}$ to  prove that 
\begin{equation} 
\overline{X_{v,i}}\vee [COUNT((X_{w,j})_{w\in C,j<i})\leq i]
\label{eq-aux1}
\end{equation}

Indeed, assuming $X_{v,i}=TRUE$ we can derive any disjunction of length $i+1$ consisting of literals of type $\overline{X_{w,j}}$,  with $w\in C$, $j<i$. This is because for all $w_1\neq w_2\in C$, $k_1\neq k_2$  $\overline{X_{w_1,k_1}}\OR \overline{X_{w_2,k_1}}$  and $\overline{X_{w_1,k_1}}\OR \overline{X_{w_1,k_2}}$ are clauses of $\Phi(G,k)$. By Lemma~\ref{foo-size} we can derive equation~(\ref{eq-aux1}). Next, simple arguments along the lines of \cite{buss-frege-php} establishes the equivalence between formulas $U\leq i$ and $\bigvee_{k=1}^{i} [U=k].$  Here $U$ is a bit vector of appropriate length to represent $i$. We use~(\ref{eq-aux1}) and this to derive~(\ref{feq}). 

Now, for every $v\in V(G^{\prime})=R$ we define a new variable $Z_{v,i}$, designed to be true iff the color of $v$ in the induced coloring on $G^\prime$ is $j$. We will enforce this by making the substitutions
\begin{equation} 
Z_{v,j}:=\bigvee_{j=1}^{i} X_{v,i}\AND [less(v)=i-j]
\end{equation}
First note that $Z$ respects the color classes of $G$: if $v_1,v_2$ have the same color in $G$ then they have the same color in $G^{\prime}$. Furthermore, the substitution does not collapse two different color classes of $G$ into a single color class in $G^{\prime}$: it simply relabels the colors of vertices in $G^{\prime}$ with elements of $1,2,\ldots, k^{\prime}$. Therefore, if $Z_{v_1,j}=Z_{v_2,j}=TRUE$ then there exists an unique $i_0$ such that $X_{v_1,i_0}=X_{v_{2},i_0}=TRUE$. 
 
We need to derive clauses $\bigvee_{t=1}^{k^{\prime}} Z_{v,t}$ as well as, for $vw\in E(G^{\prime})$, $\overline{Z_{v,j}}\OR \overline{Z_{w,j}}$. Deriving the first type of clauses is easy: we use formulas~(\ref{feq}) and $X_{v,1}\OR X_{v,2}\OR \ldots \OR X_{v,k}$. 

As for the second one, nota that all 
 clauses $\overline{X_{v,i}}\OR \overline{X_{w,i}}$ are part of $\Phi(G,k)$. Given the observation we made above and this fact, assuming $Z_{v,j}=Z_{w,j}=TRUE$ we can derive a contradiction. By Lemma~\ref{foo-size} we can, therefore, derive (with the same complexity) clause $\overline{Z_{v,j}}\OR \overline{Z_{w,j}}$.

\section{Proof of Theorem~\ref{theorem-vc}}

We use the predicate $COUNT^{n}_{k}(x_{1},x_{2},\ldots, x_{n})$ from \cite{buss-frege-php}. Formula  \\ $COUNT^{n}_{k}(x_{1},x_{2},\ldots, x_{n})$ is TRUE if and only if at least $k$ of the variables $x_{1},x_{2},\ldots, x_{n}$ are true. For every fixed $k$, $COUNT^{n}_{k}$ can be computed by polynomial size Frege proofs. 

We will define a sequence of formulas: 
\begin{enumerate} 
\item For $v\in V$, $\Phi_{v,1}(\overline{Y})= COUNT_{k}^{n-1}((Y_{v,w})_{w\neq v\in V})$. Informally, $\Phi_{v,1}$ is true in graph $G$ iff the degree of $v$ is at least $k$. 
\item For $v\in V$, $\Phi_{v,2}(\overline{X},\overline{Y})= (\bigwedge\limits_{i=1}^{k} \overline{X_{v,i}})\wedge \Phi_{v,1}(\overline{Y})\wedge \Phi_{VC}(G,k)[\overline{X},\overline{Y}]$. 
\end{enumerate} 

For every neighbor $w$ of $v$, by resolving $Y_{v,w}$ with clause $\overline{Y_{v,w}}\vee X_{v,1}\vee \ldots \vee X_{v,k}\vee X_{w,1}\vee \ldots \vee X_{w,k}.$ of $\Phi_{v,2}$ we derive clause $X_{v,1}\vee \ldots \vee X_{v,k}\vee X_{w,1}\vee \ldots \vee X_{w,k}.$ By resolving successively with $\overline{X_{v,1}} ,\ldots, \overline{X_{v,k}}$ we derive clause $X_{w,1}\vee \ldots \vee X_{w,k}.$

Formula $\bigwedge\limits_{w\in N(v)} (X_{w,1}\vee \ldots \vee X_{w,k})$ is isomorphic to the Pigeonhole Principle $PHP_{|N(v)|}^{k}$ which  has polynomial-size Frege refutations \cite{buss-frege-php}. Plugging in this proof of this statement into our argument, we conclude that that the implication $\Phi_{v,2}(X,\overline{Y})\vdash \square$ can be witnessed by polynomial size Frege proofs, hence, by Lemma~\ref{foo-size}, so does the implication $\Phi_{v,1}(\overline{Y})\wedge \Phi_{VC}(G,k)[\overline{X},\overline{Y}]\vdash \bigvee\limits_{i=1}^{k} X_{v,i}$. 

As for the second reduction rule, it is just as easy: for every vertex $v\in V$ which is isolated and every $i=1,\ldots, k$, we first derive by resolution (using negative clauses $\overline{Y_{v,w}}$ and clause $\overline{X_{v,i}}\vee (\bigvee_{w\neq v} Y_{v,w})$ unit clauses $\overline{X_{v,i}}$. We then use these clauses to resolve away every other occurrence of $X_{v,i}$ from the formula, obtaining a formula isomorphic to $\Phi(G\setminus Isolated(G),k)$. 

%\section{Proof of Theorem~\ref{theorem-vc}}

%[
%\Phi(G,k)\vdash \Phi(G\setminus \{v\},k-1),\mbox{ and }
%\]
%\[
%\Phi(G,k)\vdash \Phi(G\setminus Isolated(G),k). 
%\]

\section{Proof of Theorem~\ref{thm-edgecover}}

The soundness of the first reduction rule, $\Phi(G,k)\vdash \Phi(G\setminus Isolated(G),k)$ can be witnessed by efficient Frege proofs similar to those for the vertex cover problem. 

As for the second rule, the formula $$\Xi_{S}(G):= \wedge_{w,v\in S}\wedge_{r\in V} (Y_{v,r}\leftrightarrow Y_{w,r})$$ (where, of course, $A\leftrightarrow B$ can be equivalently rewritten as $(\overline{A}\vee B)\wedge (A\vee \overline{B})$) expresses the fact that $N[v]=N[w]$ for all $v,w\in S$. So we need to prove the soundness of the rule 
\begin{equation} 
\Phi(G,k)\wedge \Xi_{S}(G)\vdash \Phi(G^{\prime},k^{\prime}). 
\end{equation} 
Without loss of generality we will only deal with the case $N[v]\neq \emptyset$ for all $v\in S$, (the other case,  $N[v]= \emptyset$ for all $v\in S$, can be handled with minor modifications to this argument). By slightly abusing notation, we will denote by $S$ the vertex of $G^{\prime}$ obtained by contraction. Let $s \in S$ be an arbitrary vertex. 

We define substitutions: $Y_{v,w}^{\prime}:= Y_{v,w}$ for all $v,w\in G^{\prime}$, $v,w\neq S$. If, say, 
$v=S$ we define $Y_{S,w}^{\prime}:=Y_{s,w}$. Also define $X^{\prime}_{v,i}:=X_{v,i}$ for $v\neq S$, $X^{\prime}_{S,i}:=X_{s,i}$. The substitution yields a formula isomorphic to $\Phi(G^{\prime},k^{\prime})$, and the proof of the safety is basically trivial. 

To obtain the result note that the number of vertices goes down geometrically, by a ratio of $1-\frac{1}{2^{k}}$ at each step. 

\subsection{Proof of Lemma~\ref{safe-ecc}}

Let $G$ be a graph to which rules (a). (b). do not apply and which has an edge clique cover of size $k$. 
Consider an encoding $b(v)$ of every vertex $v$ on $k$ bits such that for every $v\in V$, $b(v)$ is a bit vector whose $i$'th bit is one iff $v$ is a part of the $i$'th clique. 

There must be a set of vertices $S\subset V$, $|S|\geq \frac{n}{2^k}$ such that for all $u,v\in S$, $b(u)=b(v)=b$, for some $b\in \{0,1\}^{k}$. 

If $b=0^{k}$ then, since every edge in $G$ must be covered by one of the $k$ cliques, it follows that every $v\in S$ is an isolated vertex. Hence rule (a). applies. 

If, on the other hand $b\neq 0^{k}$, say $b_{i}\neq 0$, then every $v\in S$ must belong to the $i$'th clique. Hence $S$ induces a clique in $G$.

\section{Proof of Theorem~\ref{theorem-hittingset}}

Define for $i\in U^{\prime}$ substitutions $X_{i,j}^{\prime}:=X_{i,j}$. We need to 
\begin{description} 
\item[-] derive clause $(\vee_{i\in Y}(\vee_{j=1,\ldots k} X_{i,j}))$. 
\item[-] for every $j=1,\ldots, k$, derive clauses $\vee_{i\in U^{\prime}} X_{i,j}$
\end{description} 
For the first clause we show that $\Phi(U,\mathcal{A},k) \wedge (\wedge_{i\in Y}(\wedge_{j=1,\ldots k} \overline{X_{i,j}}))\vdash \emptyset$ and then we invoke Lemma~\ref{foo-size}. To do that we first derive, for $l=1,\ldots, k+1$ (by resolving literals $\overline{X_{i,j}}$) clauses $(\vee_{i\in S_{l}\setminus Y}(\vee_{j=1,\ldots k} X_{i,j}))$. 

Substituting each variable $X_{i,j}$, where $i\in S_{j}\setminus Y$ to a new variable $Y_{i,j}$ 
yields a formula isomorphic to $PHP_{k+1}^{k}$ which has polynomial size Frege proofs \cite{buss-frege-php}. Putting all these things together we get polynomial-size Frege proofs witnessing the soundness of one step of the data reduction. 

The number of clauses drops at every reduction step by $k$, so the length of the data reduction chain is $O(n^d/k)$. 

\section{Proof of Theorem~\ref{theorem-arrow}}

First, a note about the length of the reduction chains. Formula $Arrow_{m,n}$ has $(m!)^{n+1}$ variables, all of them appearing explicitly in the formula. But $n=O(log( (m!)^{n+1}))$, so indeed a reduction chain of length $O(n)$ has length logarithmic in $|Arrow_{m,n}|$.

\subsection{Proof of Lemma~\ref{arrow-red}}

Suppose that $a,a^{\prime}\in [m]_{-B}$ and $a<_{R_{i}} a^{\prime}$ for all $i\in [m]_{-B}$. Then $a<_{R_{i}^{+B}} a^{\prime}$. By unanimity of $W$, $a<_{W(R_{1}^{+B},\ldots, R_{n}^{+B})} a^{\prime}$.
Since $a,a^{\prime}$ were arbitrary, it follows that $W_{-B}$ is unanimous.  

As for IIA, let $a,a^{\prime}\in [m]_{-B}$ and $(R_1,R_2,\ldots, R_n)$ and $(R_1^{\prime},R_2^{\prime},\ldots, R_n^{\prime})$ be preference profiles such that, for every $i=1,\ldots, n$, $R_i$ and $R_{i}^{\prime}$ agree with respect to the relative ordering of $a,a^{\prime}$. Then for every $i=1,\ldots, n$, $R_i^{+B}$ and $R_{i}^{\prime,+B}$ agree with respect to the relative ordering of $a,a^{\prime}$. By the IIA axiom for $W$, $W(R_1^{+B},R_2^{+B},\ldots, R_n^{+B})$ and $W(R_1^{\prime,+B},R_2^{\prime,+B},\ldots, R_n^{\prime,+B})$ agree with respect to the relative ranking of $a,a^{\prime}$. Hence so do $W_{B}(R_1,R_2,\ldots, R_n)$ and $W_{B}(R_1^{\prime},R_2^{\prime},\ldots, R_n^{\prime})$. 

\subsection{Propositional simulation of the reduction rules}
\begin{description} 
\item[First reduction rule: ] 
Define, for $Q\subseteq [m]$, $|Q|=m-5$, and $i=1,\ldots, n$ formulas
\[
Nondict_{-Q,i}:=\bigvee_{R\in \mathcal{R}_{-Q}} \overline{X_{R^{+Q},R^{+Q}_{i}}}
\]
(informally, formula $Nondict_{-Q,i}$ is true iff $i$ is not a dictator for $W_{-Q}$). 
\[
Unanimous_{-Q}:=\bigwedge_{a,b\in [m]_{-Q}} \bigvee_{\stackrel{R\in \mathcal{R}_{a,b,-Q}}{\pi\in S^{m}_{a,b}}} X_{R^{+Q},\pi^{+Q}}. 
\]
\[
IIA_{-Q}:= \bigwedge_{(R,R^{\prime},\pi,\pi^{\prime})\in \mathcal{R}_{-Q}} (\overline{X_{R^{+Q},\pi_{1}^{+Q}}}\OR \overline{X_{R^{\prime,+Q},\pi_{2}^{+Q}}})
\]
(where, for simplicity, we have ommitted the IIA restrictions on $R,R^{\prime},\pi_1,\pi_2$, see Definition~\ref{def-arrow}). Note that $Arrow_{m,n}=Unanimous_{-\emptyset}\AND IIA_{-\emptyset} \AND \bigwedge_{i=1}^{n} Nondict_{-\emptyset,i}$. 

We will prove that 
\begin{equation}
Arrow_{m,n}\AND \bigwedge_{i=1}^{n} Nondict_{-[6:m],i} \vdash Arrow_{5,n}
\label{arrow-1} 
\end{equation} 
and, for $i=1,\ldots, n$
\begin{equation}
Arrow_{m,n}\AND \bigvee_{i=1}^{n} \bigwedge_{R\in \mathcal{R}_{-[6:m]}} X_{R^{+[6:m]},R^{+[6:m]}_i} \vdash Arrow_{5,n}
\label{arrow-2} 
\end{equation} 
Employing tautology 
\begin{equation}
(\bigwedge_{i=1}^{n} Nondict_{-[6:m],i}) \OR \bigvee_{i=1}^{n} (\bigwedge_{R\in \mathcal{R}_{-[6:n]}} X_{R^{+[6:m]},top(R^{+[6:m]}_i)})
\label{arrow-3} 
\end{equation} 
and substitutions implicit in~(\ref{arrow-1}) and~(\ref{arrow-2}) we conclude that $Arrow_{m,n} \vdash P_1 \OR P_2$, where both $P_1,P_2$ are formulas isomorphic to $Arrow_{5,n}$. Thus we are in the framework of our metatheorem with $R=2$. Note that we do not need to prove (7) for \emph{all} $i$, but only for one, the one that is a dictator. 

We need to specify the substitutions implicit in~(\ref{arrow-1}) and~(\ref{arrow-2}). First, formalizing the mathematical argument in Lemma~\ref{arrow-red} we show that $Arrow_{m,n}\vdash IIA_{-Q}\AND Unanimous_{-Q}$. The propositional content of this implication is trivial: the clauses of $Unanimous_{-Q}$ and $IA_{-Q}$ are simply subclauses of $Arrow_{m,n}$. 

Because of this, the substitution witnessing implication~(\ref{arrow-1}) is quite simple: it replaces a restricted variable $X_{R,\pi}^{\prime}$ of $Arrow_{5,n}$ with the variable $X_{R^{+[6;m]},\pi^{+[6:m]}}$ of $Arrow_{m,n}$. 

As for~(\ref{arrow-2}), define, for all $c< d\in [5]$ and $i=1,\ldots, n$, formulas
\begin{equation} 
Witness_{c,d,i}:= \bigvee_{\stackrel{R\in \mathcal{R}:R_{i}\in S_{c,d}^{m}}{\pi \in S_{d,c}^{m}}} X_{R,\pi}
\end{equation} 
Informally, formula $Witness_{c,d,i}$ is true when pair $(c,d)$ acts as a witness that agent $i$ is not a dictator for $W$, since $c<_{R_i} d$ but $d<_{W(R)}c$. 

A next step is to prove that $$Arrow_{m,n}\AND \bigwedge\limits_{R\in \mathcal{R}_{-[6:m]}} X_{R^{+[6:m]},R^{+[6:m]}_i} \vdash  \bigvee\limits_{c<d\in [5]} Witness_{c,d,i}.$$ This is easy: we use literals $X_{R^{+[6:m]},R^{+[6:m]}_i}$ to prove (by resolution) 
\begin{equation} 
Arrow_{m,n}\AND \bigwedge\limits_{R\in \mathcal{R}_{-[6:m]}} X_{R^{+[6:m]},R^{+[6:m]}_i} \AND \bigwedge \limits_{c<d\in [5]} \bigwedge\limits_{\stackrel{R\in \mathcal{R}:R_{i}\in S_{c,d}^{m}}{\pi \in S_{d,c}^{m}}} \overline{X_{R,\pi}}\vdash \square
\end{equation} 
(the last conjunction negates formulas $Witness_{c,d,i}$) and then invoke Lemma~\ref{foo-size} to get a proof of the same length of the implication we claimed. 

Now we prove, for all $j=1,\ldots, n$ that 
\[
Arrow_{m,n}\AND Witness(c,d,i)\vdash Nondict_{-V,j}, 
\]
where $V$ is defined as in the proof of Lemma~\ref{safe-a}. 

The proof of the implication for $j=i$ uses unit literal $X_{R,R_j}$  (negation of one from $Nondict_{-V,i}$) and $\overline{X_{R,\pi}}\OR \overline{X_{R_{-V}^{+V},\pi_2}}$ (part of the IIA part of $Arrow_{m,n}$) to derive clauses $\overline{X_{R_{-V}^{+V},\pi_2}}$ for all $\pi_2$ that rank $c,d$ in a different way than $\pi$. Resolving away these literals from clause $\bigvee\limits_{\pi\in S_m} X_{R^{+V}_{-V},\pi}$ (part of $Arrow_{m,n}$) derives clause $Nondict_{-V,i}$. 

As for the case $j\neq i$, we want to show that 
\begin{equation} 
 Arrow_{m,n} \AND Witness(c,d,i)\AND \bigwedge\limits_{R\in \mathcal{R}_{[5]}} X_{R^{+[6:m]},R^{+[6:m]}_j}\vdash \square
 \label{der-nondict} 
 \end{equation} 
 By Lemma~\ref{foo-size} this will imply $Nondict_{-V,j}$. 

Let $R=(R_1,R_2,\ldots, R_{n})$ be a profile on $[5]$ such that $a<_{R_1}b$ but $b<_{R_2}a$ ($a,b\in [5]$ are defined as in Lemma~\ref{safe-a}). 

Combining the derivations of $Nondict_{-V,j}$ in~(\ref{der-nondict}) with the ones of \\ $Unanimous_{-V}$ and $IIA_{-V}$ (outlined before), plus a bijective identification of of $[m]_{-V}$ and $[5]$ yields a substitution that proves $Arrow_{5,n}$, completing the proof of~(\ref{arrow-2}). 

\item[Second reduction rule:] We refer to Lemma 2 of \cite{tang2009computer} for (mathematical) details of the reduction.  What is important is that the soundness of the statement that at least one of $W_{1,2},W_{2,3},W_{1,3}$ is non-dictatorial is established by a case-by-case analysis. It is first proved that it cannot be that all these functions have the same dictator. Then it is established that if $i$ is the dictator of $W_{1,2}$, $j$ the dictator of $W_{1,3}$, $k$ the dictator of $W_{2,3}$ then $i\in \{2,3\}$, $j\in \{2,3\}$, $k\in \{1,3\}$. For all eight possible cases for triplets $(i,j,k)$ we obtain a contradiction: either we explicitly provide a profile $R$ showing that triplet $(i,j,k)$ cannot represent the set of dictators for the three function, or we employ an argument similar to the one in the case $i=j=k$.  
\end{description} 
\section{Proof of Theorem~\ref{theorem-gs}}

First, it is not obvious that, as formulated in the paragraph, $W_{-B}$ is well-defined. The reason is that $W_{-B}(R)$ invokes $W$ on profile $R^{+B}$, and it is not obvious that if $R$ is a profile on $[m]_{-B}$ the the outcome of $W$ is an element of $[m]_{-B}$, as needed by the definition. 

Suppose that $W_{-B}(R)=W(R^{+B})\in B$ for some profile $R$ on $[m]_{-B}$. We claim that $W(S)\in B$ for every profile $S=(S_1,S_2,\ldots, S_n)$, contradicting the hypothesis that $W$ is onto.  Indeed, if $W_{-B}(R)=W(R^{+B}) \in B$ then $W_{-B}(R_{-1},S_1)=W(R_{-1}^{+B},S_1^{+B})\in B$, otherwise agent 1 would have an opportunity to manipulate at profile $R^{+B}$ by misrepresenting its preference as $S_{1}^{+B}$. Applying this argument inductively for agents $2,3, \ldots, n$ (replacing $R_i$ by $S_i$) we infer that $W(S)\in B$, which is what we claimed. 

\subsection{Proof of Lemma~\ref{gs-red}}

Suppose there exists some profile $R$ and agent $i\in[n]_{-B}$ such that $i$ could manipulate $W_{-B}(R)$ by misrepresenting its profile as $R^{\prime}_{i}$. This means that $i$ could manipulate $W$ on profile $R^{+B}$ by misrepresenting its profile as $R^{\prime,+B}_{i}$, contradicting the fact that $W$ is strategy-proof. Hence $W_{-B}$ is strategy-proof. 

Suppose now that $a\in [m]_{-B}$. Since $W$ is onto, there must exist a profile $R$ such that $W(R)=a$. 
Consider the profile $R^{\prime}_{i}$ that modifies $R$ by moving $a$ to the top of preference profile $R_i$. We claim that $W(R^{\prime}_{i})=a$. Indeed, if this was not the case then $i$ could manipulate on profile $R^{\prime}_{i}$ by misrepresenting its preferences. Continuing the argument inductively for all agents we infer that if $\overline{R}$ is the profile that modifies $R$ by moving $a$ to the top of all profile preferences then $W(\overline{R})=a$. This means that $W$ is unanimous, hence $W_{-B}$ also is. But then there is a profile $\underline{R}$ such that $W_{-B}(\underline{R})=a$: simply make $a$ the top of all profiles.  Since $a$ was arbitrary, it follows that $W_{-B}$ is onto.   

\subsection{Proof of safety of reduction rule a.} 
\label{sect} 

We show that there exists $T\subseteq [m]$, $|T|=m-3$ such that $W_{-T}$ is non-dictatorial. Together with Lemma~\ref{gs-red} this establishes the safety of rule a. 

\textbf{Step 1}. Given two different sets $T_1,T_2$ of size $m-3$, $|T_1\cap T_2|=m-4$ (in other words, $|\overline{T_1}|=|\overline{T_2}|=3$, $|\overline{T_1}\cap \overline{T_2}|=2$), we show that functions $W_{-T_1},W_{-T_2}$ must have the same dictator, if they have one. 

Suppose, indeed, that $W_{-T_1}$ has dictator $i$, $W_{-T_2}$ has dictator $j\neq i$. 
Let $d\in T_2\setminus T_1$, $c\in T_1\setminus T_2, a,b\in \overline{T_1}\cap \overline{T_2}$. Define profiles 
\[
R_s= a<b< \ldots < sorted(T_2) < c, \mbox{ for }s\neq j
\]
\[
R_j= b<a< \ldots < sorted(T_2) < c. 
\]
$W(R)=a$, since $i$ is a dictator for $W_{-T_1}$. Now, if we replace $R_i$ by 
\[
R_{i}^{\prime}=a<b< \ldots < c < sorted(T_2). 
\]
obtaining profile $R^{\prime}$, then $W(R^{\prime})=a$, otherwise agent $i$ could manipulate by reporting profile $R_i$ instead. We continue changing iteratively profiles $R_s, s\neq j$ to $R_s^{\prime}=a<b< \ldots < c< sorted(T_2)$, one profile at a time, until all profiles $R_s$ except $R_j$ have been replaced by $R_s^{\prime}$. Call this profile $\underline{R}$. That is 
\[
\underline{R}_s= a<b< \ldots < c < sorted(T_2) \mbox{ for }s\neq j
\]
\[
\underline{R}_j= b<a< \ldots < sorted(T_2) < c. 
\]

Since no agent $s\neq j$ had an opportunity to manipulate, it must be that $W(\underline{R})=a$. Consider now profile $\underline{R}_j^{\prime}= b<a< \ldots < c < sorted(T_2)$. Since $W_{-T_2}$ has agent $j$ as a dictator, $$W(\underline{R}_1,\ldots, \underline{R}_{j}^{\prime},\ldots, \underline{R}_{n})= W_{-T_2}(\underline{R}_{1,-T_2},\ldots, \underline{R}^{\prime}_{j,-T_2}, \underline{R}_{n,-T_2})=b.$$ So agent $j$ has an opportunity to manipulate at profile $(\underline{R}_1,\ldots, \underline{R}_{j}^{\prime},\ldots, \underline{R}_{n})$ by reporting instead $\underline{R}_{j}^{\prime}$.

\textbf{Step 2}. Either there exists a set $T$ of size $m-3$ such that $W_{-T}$ is not dictatorial, or all functions $W_{-T},|T|=m-3$ must have the same dictator. Indeed, we can "interpolate" between any two sets of cardinality $m-3$ by a sequence of sets falling under step 1. 

\textbf{Step 3}. We show that it is not possible that all functions $W_{-T}$, $|T|=m-3$ have the same dictator $i$. Since $W$ is not dictatorial, there exists a profile $R$ such that $b=W(R)$ is different from $a=top(R_i)$. Let $T\subseteq [m], |T|=m-3,a,b\not \in T$ and consider a profile $R^{\prime}$ that modifies $R$ by moving $b,sorted(T)$ to the bottom of all preferences (in this order), that is $R^{\prime}=(R^{+b}_{-b})^{+sorted(T)}_{-sorted(T)}$.

We have $W(R^\prime)=W_{-T}(R_{-T}^{\prime})=top(R_{i}^{\prime})=a$, since $W_{-T}$ has $i$ as dictator and $top(R^{\prime}_{i})=a$. Let us create a path between $R$ and $R^{\prime}$ by changing one profile $R_{s}$ at a time to $R_{s}^{\prime}$, the last move being $R_i$. 

The value of $W$ does not change at any step $s$, since $W_{-T}$ has $i$ as dictator, and the relative orders of elements in $[m]_{-T}$ does not change at any profile as a result of a change $R_s\rightarrow R_{s}^{\prime}$, $s\neq i$, or agent $s$ would have an opportunity of manipulation at one of the two profiles, using $R_s,R_{s}^{\prime}$, whichever yields a result ranked lower in $R_s,R_{s}^{\prime}$. But this yields a contradiction, since $W(R)=b$ and $W(R^{\prime})=a\neq b.$ 

\subsection{Propositional simulation of the reduction rules}

\textbf{First reduction rule}

Define, for $Q\subseteq [m]$, $|Q|=m-5$, and $i=1,\ldots, n$ formulas
\[
Nonmanip_{-Q,\pi,i}:=\bigwedge_{\stackrel{R\in \mathcal{R}_{-Q}}{o\in  [m]_{-Q}}} \overline{X_{R^{+Q},o}} \OR (\bigvee_{o^{\prime}\in pr(i,o,R^{+Q})} X_{s(i,R^{+Q},\pi),o^{\prime}})
\]

Just as in the case of the first reduction rule of Arrow's theorem, proving propositionally that $W$ is non-manipulable implies that $W_{-Q}$ is non-manipulable is easy, since clauses of $Nonmanip_{-Q,\pi,i}$ are a subset of those in $Nonmanip_{-\emptyset,\pi,i}$, which are part of $GS_{m,n}$. 

Similarly, define, for $Q\subseteq [m]$, $|Q|=m-5$, and $i=1,\ldots, n$ formulas
\[
Onto_{-Q}:=\bigwedge_{o\in  [m]_{-Q}} \bigvee_{R\in \mathcal{R}_{-Q}} X_{R^{+Q},o}. 
\]
$Onto_{-\emptyset}$ is part of $GS_{m,n}$.  
We are left to giving details about proving that $Onto_{-Q}$ follows from $GS_{m,n}$, and about the 
propositional formalization of the soundness of reduction rule a in section~\ref{sect}. We defer posting these details to the final version of this paper, to be posted online. 

\textbf{Second reduction rule} 

We refer to Lemma 2 of \cite{tang2008computer} for (mathematical) details of the reduction.  It is first proved that if $i$ is the dictator of $W_{1,2}$, $j$ the dictator of $W_{1,3}$, $k$ the dictator of $W_{2,3}$ then $i=2$, $j=3$, $k=3$. 

Propositionally, this amounts to proving that for every $r\neq 2$, 
\begin{equation} 
\bigvee_{R\in \mathcal{R}_{[2]}}  \overline{X_{R,top(R_r)}}
\end{equation} 
and similarly for functions $W_{1,3},W_{2,3}$. 

Indeed, $Nondict_{\emptyset,r}$ is part of $Arrow_{m,n}$. 
For any literal $\overline{X_{R,top(R_r)}}$ part of $Nondict_{\emptyset,r}$, we resolve it against clause 
$\bigvee_{o\in [m]} X_{R,o}$  to derive $\bigvee_{o\in [m],o\neq top(R_r)} X_{R,o}$. 

For every literal $X_{R,o}$ in this clause, we resolve it against the formula \\ $Nonmanip_{-\emptyset,\pi,1}$ to derive $\bigvee_{o^{\prime}\in pr(1,o,R)} X_{s(1,R,\pi),o^{\prime}}$. Now take 
$\pi=(R_1)_{-o}^{o+}$. Resolving literal $X_{s(1,R,(R_1)_{-o}^{o+}),o^{\prime}}$ against clause 
  $\overline{X_{s(1,R,(R_1)_{-o}^{o+}),o^{\prime}}}\OR \overline{X_{s(1,R,(R_1)_{-o}^{o+}),o^{\prime}}}$
(part of $GS_{m,n}$) 

Finally, the proof in \cite{tang2008computer} proceeds to obtain a contradiction from the claim that $2$ is the dictator of $W_{1,2}$, $3$ is the dictator of $W_{1,3}$, $W_{2,3}$ by constructing a new profile $R^{\prime}$ and exibiting, for every possible value of $W(R^{\prime})$ a contradiction to manipulability. This can be simulated propositionally: each case amounts to a literal in a disjunction of clause $\bigvee_{\lambda\in [m]} X_{R^{\prime},\lambda}$ of $GS_{m,n}$. Each assumption can be simulated by a resolution proof, and to obtain a resolution proof of the contradiction we use Lemma~\ref{foo-size} to derive negations of these literals, and then obtain a contradiction by resolving all these negated literals against clause $\bigvee_{\lambda\in [m]} X_{R^{\prime},\lambda}$. 
\end{document}